\documentclass[a4paper,UKenglish,cleveref,autoref,nameinlink,thm-restate]{lipics-v2021}

\usepackage[dvipsnames]{xcolor}
\usepackage{mathrsfs}
\usepackage{complexity}
\usepackage{amsmath,amsfonts,amssymb}
\usepackage{mathtools}
\usepackage{todonotes}
\usepackage{hyperref}
\usepackage{booktabs}
\usepackage{subcaption}
\usepackage{xspace}
\newcommand{\size}{\texttt{size}}
\newcommand{\degree}{\texttt{degree}}
\newcommand{\softcore}{\texttt{softcore}}
\newcommand{\degeneracy}{\texttt{degeneracy}}
\newcommand{\combo}{\texttt{combo}}
\newcommand{\ttt}{\texttt{TTT}}
\newcommand{\hkmn}{\texttt{HKMN}}

\usepackage{algorithm}
\usepackage[noend]{algpseudocode}

\newcommand{\bigO}{\mathcal{O}}

\definecolor{lipicsblue}{rgb}{0.08235294118,0.3098039216,0.537254902}
\definecolor{defblue}{rgb}{0.121,0.47,0.705}
\definecolor{orangered}{rgb}{1,0.271,0}
\definecolor{darkred}{rgb}{0.7,0.1,0}
\DeclareTextFontCommand{\emph}{\color{defblue}\em}

\hypersetup{colorlinks=true,
    linkcolor=darkred,
    anchorcolor=lipicsblue,
    citecolor=lipicsblue,
    filecolor=lipicsblue,
    menucolor=lipicsblue,
    urlcolor=lipicsblue,
    bookmarksopen=true,
    bookmarksopenlevel=2,
    bookmarksnumbered=true,
    plainpages=false,
}

\hideLIPIcs
\nolinenumbers

\title{Engineering Algorithms for $\ell$-Isolated Maximal Clique Enumeration}

\author{Marco {D'Elia}}{Roma Tre University, Rome, Italy}{marco.delia@uniroma3.it}{https://orcid.org/0009-0008-6266-3324}{}

\author{Irene Finocchi}{Luiss Guido Carli, Rome, Italy}
{finocchi@luiss.it}{https://0000-0002-6394-6798}{}

\author{Maurizio Patrignani}{Roma Tre University, Rome, Italy}{maurizio.patrignani@uniroma3.it}{https://orcid.org/0000-0001-9806-7411}{}

\authorrunning{M. D'Elia et al.}

\titlerunning{Engineering Algorithms for $\ell$-Isolated Maximal Clique Enumeration}

\Copyright{}



\category{} %optional, e.g. invited paper

\relatedversion{} %optional, e.g. full version hosted on arXiv, HAL, or other respository/website
\funding{This research was supported in part by MUR PRIN Projects EXPAND and NextGRAAL [grant numbers 2022TS4Y3N and 2022ME9Z78].}

\acknowledgements{We thank Robert Bredereck and Manuel Sorge for providing the code used in~\cite{DBLP:journals/tcs/HuffnerKMN09}. We thank Tommaso Caiazzi and Simone Romanelli for interesting conversations.}

\begin{document}

\maketitle

\begin{abstract}
Maximal cliques play a fundamental role in numerous application domains, where their enumeration can prove extremely useful. Yet their sheer number, even in sparse real-world graphs, can make them impractical to be exploited effectively.  
To address this issue, one approach is to enumerate $\ell$-isolated maximal cliques, whose vertices have (on average) less than $\ell$ edges toward the rest of the graph. By tuning parameter~$\ell$, the degree of isolation can be controlled, and cliques that are overly connected to the outside are filtered out. 
Building on Tomita {\em et al.}’s very practical recursive algorithm for maximal clique enumeration, we propose four pruning heuristics, applicable individually or in combination, that discard recursive search branches that are guaranteed not to yield $\ell$-isolated maximal cliques.  
Besides proving correctness, we characterize both the pruning power and the computational cost of these heuristics, and we conduct an extensive experimental study comparing our methods with Tomita’s baseline and with a state-of-the-art approach.  
Results show that two of our heuristics offer substantial efficiency improvements, especially on real-world graphs with social network properties.

\end{abstract}

\section{Introduction}\label{sec:introduction}

Cliques capture dense, highly-connected network structures. Their enumeration is important in many application domains, such as Web mining, bioinformatics, computational chemistry, social and financial network analysis. In particular, enumerating {\em maximal} cliques — those that cannot be extended by another vertex — helps reveal the graph's structure by uncovering tightly connected groups of vertices that share meaningful relationships.

From both theoretical and practical perspectives, the problem presents significant challenges. The set of maximal cliques includes indeed also {\em maximum} cliques, but computing a maximum clique is a well-known \NP-hard problem~\cite{DBLP:conf/coco/Karp72}, as well
as \W[1]-hard~\cite{DBLP:series/mcs/DowneyF99,downey2012parameterized}, and hard to approximate within a factor of $n^{1-\epsilon}$~\cite{hastad1999clique,DBLP:conf/stoc/Zuckerman06}. %,DBLP:conf/stoc/Zuckerman06}.
In general, the number of maximal cliques in an $n$-vertex graph can be exponential in $n$~\cite{Moon1965OnCI}. Hence, not only their enumeration requires considerable time, but navigating such a large number of structures is also difficult. Imposing additional constraints may pinpoint a subset of domain-relevant structures, substantially reduce the number of enumerated cliques, and enable more effective utilization. We refer to~\cite{DFP25} for a taxonomic survey of different summarization approaches.

% In contexts where cliques with weak connections to the rest of the graph are more meaningful than those strongly intertwined with it, the concept of {\em $\ell$-isolation} may prove particularly useful. First introduced by Ito {\em et al.} in~\cite{DBLP:conf/esa/ItoIO05,DBLP:journals/talg/ItoI09}, $\ell$-isolation requires that the number of edges leaving a clique of size $k$ does not exceed $\ell\cdot k$.
% An enumeration algorithm was first proposed in~\cite{DBLP:conf/esa/ItoIO05} to list all $\ell$-isolated maximal cliques in $\FPT$ time with respect to $\ell$. Since then, several papers have investigated different variants of clique isolation principles and computationally evaluated their effectiveness.

In contexts where cliques with weak connections to the rest of the graph are more meaningful than those strongly intertwined with it, the concept of {\em $\ell$-isolation} may prove particularly useful. First introduced in~\cite{DBLP:conf/esa/ItoIO05,DBLP:journals/talg/ItoI09},
% ~\cite{DBLP:journals/talg/ItoI09},
$\ell$-isolation requires that the number of edges leaving a clique of size $k$ does not exceed $\ell\cdot k$.
In the same works, the authors proposed an \FPT\ algorithm in the isolation factor $\ell$ to enumerate all $\ell$-isolated maximal cliques. 
It was also shown that $\ell$ must be constant to achieve linear-time enumeration, since otherwise one can construct instances admitting a superlinear number of $\ell$-isolated maximal cliques.

In~\cite{DBLP:journals/tcs/KomusiewiczHMN09}, two additional isolation principles were introduced, namely max- and min-isolation. 
%and then notion of maximal $\ell$-isolated cliques was considered, i.e., $\ell$-isolated cliques that are not contained in larger $\ell$-isolated cliques. 
A systematic comparison of different isolation notions, complemented by an experimental evaluation on synthetic and financial networks, was presented in~\cite{DBLP:journals/tcs/HuffnerKMN09}, underscoring the importance of isolated clique enumeration for network analysis.
%The study of isolation has been extended in several directions. 
Isolation concepts have been also explored in temporal networks~\cite{DBLP:journals/netsci/MolterNR21,DBLP:conf/gpc/GaoHYYM20}, and applied to practical scenarios such as web structure mining~\cite{DBLP:journals/ieicet/UnoOU07}, hierarchical representations of scale-free networks~\cite{DBLP:journals/jgaa/ShigezumiUW11}, and network contraction via isolated cliques~\cite{DBLP:journals/ieicet/UnoO13}. Alternative isolation models have also been proposed, including $\tau$-isolation based on coreness~\cite{DBLP:conf/mldm/OkuboHT16}, as well as variants for stars~\cite{DBLP:conf/waw/UnoOU06}, bicliques~\cite{DBLP:journals/ijcm/AlamgirKH17}, and even approaches inspired by quantum algorithms~\cite{DBLP:conf/opodis/GallNNO24}.

%The focus of this work is on $\ell$-isolated maximal cliques, as defined in~\cite{DBLP:journals/talg/ItoI09}.
%in~\cite{DBLP:conf/esa/ItoIO05,DBLP:journals/talg/ItoI09}.

\medskip
\noindent{\bf Our contribution.} In this paper, building upon the well-known Bron-Kerbosch algorithm~\cite{DBLP:journals/cacm/BronK73} for maximal clique enumeration, improved by a pivoting strategy introduced by Tomita {\em et al.}~\cite{DBLP:conf/cocoon/TomitaTT04,DBLP:journals/tcs/TomitaTT06},
% ~\cite{DBLP:journals/tcs/TomitaTT06}, 
we first design four pruning heuristics aimed at reducing the search space, proving their correctness. The heuristics, that can be applied individually or in combination, allow the enumeration algorithm to focus only on promising parts of the search tree, discarding branches that are guaranteed not to yield $\ell$-isolated maximal cliques.
Following a theoretical characterization of the heuristics' pruning power and computational cost, which turn out to be inversely proportional, we conduct an extensive experimental study assessing their practical impact. Our analysis addresses both real and synthetic instances and contributes along several dimensions:
\begin{itemize}
    \item It sheds light on the properties of $\ell$-isolated cliques relative to the full set of maximal cliques in real-world graphs. This can be particularly useful for guiding the choice of parameter $\ell$.
    
    \item It gives insights into the relative power of our pruning strategies and of their combinations, identifying the most effective ones.
    
    \item It compares our approach with a carefully tuned baseline algorithm~\cite{DBLP:conf/cocoon/TomitaTT04,DBLP:journals/tcs/TomitaTT06} 
% algorithm~\cite{DBLP:journals/tcs/TomitaTT06} 
    and with a top-down state-of-the-art approach to the problem~\cite{DBLP:journals/tcs/HuffnerKMN09}. Overall, two of our variants offer substantial efficiency improvements over the state of the art, by as much as a factor of 4,  especially when used on real-world graphs that exhibit social network properties.
\end{itemize} 
\noindent  

%we focus first on combinatorial aspects of $\ell$-isolated maximal cliques, by providing new insights in terms of their cardinality with respect to the size of the graph. 

%    \item We prove that when the isolation factor $\ell$ is constant, the number of $\ell$-isolated maximal cliques (under all three isolation principles) is linear in the number of vertices $n$, i.e., in $\bigO(n)$. This result strengthens the results of~\cite{DBLP:conf/esa/ItoIO05,DBLP:journals/talg/ItoI09}, where the authors consider the case in which the isolation factor is non-constant or satisfies $\ell \in \omega(\log n)$, and show that the number of isolated cliques can be superlinear or even superpolynomial in such cases.

\smallskip
\noindent{\bf Paper organization.} The rest of this paper is organized as follows. 
% After discussing the most closely related works in~\cref{sec:related},  
We provide preliminary definitions and a description of the baseline enumeration algorithm in \cref{sec:preliminaries}. The pruning strategies are introduced and theoretically analyzed in~\cref{sec:enumeration}. Our experimental setup and the main outcomes of our analysis are discussed in~\cref{sec:setup} and~\cref{sec:experiments}, respectively. We summarize our results and outline open problems in \cref{sec:conclusions}.

\section{Preliminaries}
\label{sec:preliminaries}

In this section we first introduce preliminary definitions and properties of $\ell$-isolation (\cref{ss:l-isolation}), and then describe the maximal clique enumeration algorithm underlying our approach (\cref{ss:mce}).
Throughout the paper we consider finite simple graphs. For a graph $G$, let $V=V(G)$ and $E=E(G)$ denote its vertex and edge sets, with $n = |V|$ and $m = |E|$ their respective sizes. For a vertex $v$, let $\delta(v)$ denote its degree and $N(v) = \{u \mid (u,v) \in E\}$ its set of neighbors. Clearly, $\delta(v)=|N(v)|$. 
For any vertex subset $I \subseteq V$, let $G[I]$ denote the subgraph of $G$ induced by $I$, with vertex set $V[I] = I$ and edge set $E[I] = \{(u,v) \in E \mid u,v \in I\}$ consisting of edges with both endpoints in $I$.
For a vertex $v \in I$, let $\delta_e(v,I)$ denote its {\em external degree with respect to $I$}, i.e.,
$\delta_e(v,I) = |\{(u,v) \in E \mid u \notin I\}|$.
The {\em external degree of $I$}, denoted as $\delta_e(I)$, is then the number of edges with one endpoint in $I$ and the other in $V \setminus I$, i.e.,

$$
\delta_e(I) = \sum_{v\in I}\delta_e(v,I) =|\{(u,v) \in E \mid v\in I~\mbox{and}~u \in V\setminus I\}|
$$
Moreover, given $I,L \subseteq V$ with $I \cap L = \varnothing$, the {\em external degree of $I$ neglecting $L$}, denoted as $\delta_e(I,L)$, is the number of edges with one endpoint in $I$ and the other in $V \setminus (I \cup L)$.

\subsection{$\ell$-isolated cliques}
\label{ss:l-isolation}

A {\em clique} is a complete subgraph of $G$, and is {\em maximal} if it is not contained in any larger clique. The concept of {$\ell$-isolation} has been first introduced in~\cite{DBLP:journals/talg/ItoI09}:
% ~\cite{DBLP:conf/esa/ItoIO05,DBLP:journals/talg/ItoI09}:
% \begin{definition}[$\ell$-isolated clique]\label{def:l-isolation}
% For any integer value $\ell>0$, a maximal clique $C$ of size $k$ is {\em $\ell$-isolated} in $G$ if it has less than $\ell \cdot k$ outgoing edges, i.e., its external degree is such that $\delta_e(C) < \ell \cdot k$.
% \end{definition}
for any integer value $\ell>0$, a set $I \subseteq V$ of vertices $G$ of size $|I|=k$ is {\em $\ell$-isolated} in $G$ if it has less than $\ell \cdot k$ outgoing edges, i.e., $\delta_e(I) < \ell \cdot k$.
In particular, we are interested into maximal cliques that are also $\ell$-isolated ($\ell$-isolated maximal cliques).
We notice that $I$ may contain vertices with more than $\ell$ external edges: the property of $\ell$-isolation holds on average over the vertices of $I$, implying that the average external degree of vertices in $I$ is less than $\ell$\footnote{A stronger notion of {\em max-$\ell$-isolation} requires every $v \in I$ to have fewer than $\ell$ external edges, but will not be considered in this paper. Algorithmic techniques similar to those presented here could nevertheless be applied.}. 

%Clearly, max-$\ell$-isolatedness implies avg-$\ell$-isolatedness, which in turn implies min-$\ell$-isolatedness, but not vice versa. Given a $n$-vertex graph $G$, we let $\cC$, $\cCmax$, and $\cCmin$ be the set of $\ell$-isolated, max-$\ell$-isolated, and min-$\ell$-isolated maximal cliques of $G$, respectively. Clearly, $|\cCmin| \geq |\cC| \geq |\cCmax|$.

% \begin{observation}\label{obs:degree}\cite{DBLP:conf/esa/ItoIO05,DBLP:journals/talg/ItoI09}
%     Let $C$ be an $\ell$-isolated maximal clique. Then there exists $v \in C$ such that $\delta_e(v,C) < \ell$. Among the vertices satisfying this inequality, we select one with minimum degree and call it the {\em pivot of $C$}, denoted by~$p(C)$.
% \end{observation}
% The above observation implies, in particular, that the pivot of a $1$-isolated maximal clique $C$ has external degree $0$, i.e., $\delta_e(p(C),C) = 0$. %\cref{obs:degree} trivially holds also for min-isolatedness and max-isolatedness.
%
We remark that the property of $\ell$-isolation is monotone. Namely, if $\ell_1>\ell_2>0$, then any $\ell_2$-isolated set is also $\ell_1$-isolated.
% \begin{observation}\label{obs:monotone}
%     Let $\ell_1>\ell_2>0$ be two positive integers. Any $\ell_2$-isolated set is also $\ell_1$-isolated.
% \end{observation}
This follows immediately from              
%\cref{def:l-isolation}, 
the definition of $\ell$-isolation, since for an $\ell_2$-isolated set $I$ we have $\delta_e(I) <  \ell_2 \cdot k<\ell_1 \cdot k$ when $\ell_1>\ell_2$. Hence, when searching for $\ell$-isolated maximal cliques, the larger is the value of $\ell$, the weaker is the isolation
of the cliques relative to the rest of the graph,
and the larger is the number of $\ell$-isolated cliques in the graph. %In the extreme case $\ell \geq n$ and the set of $n$-isolated maximal cliques coincides with the set of all maximal cliques of $G$.
If $\ell$ is high enough, the set of $\ell$-isolated maximal cliques coincides with the set of all maximal cliques of $G$.

\subsection{Maximal clique enumeration with pivoting}
\label{ss:mce}

Our approach hinges upon the recursive backtracking procedure for maximal clique enumeration outlined in~\cref{alg:bk}. 
The procedure extends the well-known algorithm by Bron and Kerbosch~\cite{DBLP:journals/cacm/BronK73} with a pivoting strategy introduced in~\cite{DBLP:conf/cocoon/TomitaTT04,DBLP:journals/tcs/TomitaTT06},
% ~\cite{DBLP:journals/tcs/TomitaTT06}
aimed at reducing the time spent on non-maximal cliques.
Three sets of vertices are maintained by  procedure {\small\sc ProcMCE}: $C$ represents the clique under construction, $P$ contains candidate vertices that can potentially extend $C$, and $X$ holds vertices that should no longer be considered, as they are already part of maximal cliques obtained from $C$ or its subsets.
% \begin{itemize}
%     \item $C$ represents the clique under construction.
%     \item $P$ contains candidate vertices that can potentially extend $C$.
%     \item $X$ holds vertices that should no longer be considered, as they are already part of maximal cliques obtained from $C$ or its subsets.
% \end{itemize}
%
As an invariant property, vertices in $P$ and $X$ are adjacent to all vertices in $C$. The  idea is to add at each step a vertex, selected from the candidates in $P$, expanding the current clique $C$ until it becomes maximal. At the first invocation, $C=X=\varnothing$ and $P=V$.
In procedure {\small\sc ProcMCE}, base cases occur when $P$ is empty: then $C$ cannot be further expanded and is returned only if $X=\varnothing$, ensuring its maximality.
Otherwise, a pivot vertex $v_p$ is chosen so as to minimize recursive calls: {\small\sc ProcMCE} is called on each $v \in P$ not adjacent to $v_p$ (possibly including $v_p$). Correctness follows since, for each neighbor $w$ of $v_p$, any maximal clique containing $C \cup \{w\}$ will still be reported by recursing on either $C \cup \{v_p\}$ or $C \cup \{u\}$, where $u \in P$ is a non-neighbor of $v_p$. 
Recursive calls update $(C,P,X)$ to $(C \cup \{v\},\, P \cap N(v),\, X \cap N(v))$, ensuring vertices in $P\cup X$ to remain adjacent to $C$. Afterwards, $v$ is removed from further consideration.

% \begin{algorithm}
% \caption{TomitaTanakaTakahashi}\label{alg:bk}
% % \caption{TTT}\label{alg:bk}
% \SetKwFunction{ProcMCE}{ProcMCE}
% \KwData{Graph $G(V,E)$}
% \KwResult{$\mathcal{M}(G)$}
% \ProcMCE{$\varnothing,V,\varnothing$}
% \BlankLine
% \Fn{\ProcMCE{$C,P,X$}}{
%     \If {$P = \varnothing$ and $X = \varnothing$}
%         {\Return $C$\label{alg:bk:return}\;}
%     choose a pivot vertex $v_p$ in $P \cup X$\ with highest $|N(v_p) \cap P|$\;
%     \ForEach{$v \in P \setminus N(v_p)$}
%         {
%         \ProcMCE{$C \cup \{v\}, P \cap N(v), X \cap N(v)$}\;
%         $P \gets P \setminus \{v\}$\;
%         $X \gets X \cup \{v\}$\;
%     }
% }
% \end{algorithm}

\begin{algorithm}[t]
\caption{TomitaTanakaTakahashi}
\label{alg:bk}
\begin{algorithmic}[1]
\State \textbf{Input:} Graph $G = (V, E)$
\State \textbf{Output:} The set $\mathcal{M}(G)$ of all maximal cliques of $G$\\
{\small\sc ProcMCE}$(\varnothing,V,\varnothing)$\\
\Procedure{ProcMCE}{$C, P, X$}
    \If{$P = \varnothing$ and $X = \varnothing$} 
    \State {\bf yield} $C$ \label{alg:bk:return}
    \State \Return
    % \Return $C$ \label{alg:bk:return}
    \EndIf
    \State Choose a pivot vertex $v_p$ in $P \cup X$ with highest $|N(v_p) \cap P|$
    \ForAll{$v \in P \setminus N(v_p)$}
        \State \Call{ProcMCE}{$C \cup \{v\}, P \cap N(v), X \cap N(v)$} \label{alg:bk:recursive-call}
        \State $P \gets P \setminus \{v\}$
        \State $X \gets X \cup \{v\}$
    \EndFor
\EndProcedure
\end{algorithmic}
\end{algorithm}

When invoked on a graph $G$, \cref{alg:bk} implicitly explores a search tree $\mathcal{T}_G$. Each node of $\mathcal{T}_G$ corresponds to a recursive call and is uniquely associated with a triple $(C,P,X)$. The leaves of $\mathcal{T}_G$ are nodes of the form $(C, \varnothing, X)$. If $X=\varnothing$, the leaf corresponds to a maximal clique that is enumerated. Otherwise, the maximality condition is not satisfied (vertices in $X$ could be added to $C$), the clique is discarded, but no recursive call needs to be done.

\begin{figure}[bt]
    \captionsetup[subfigure]{justification=centering}
    \centering
    \begin{subfigure}{0.24\textwidth}
    \centering
    \includegraphics[page=1, trim=100 260 430 115, clip, width=\textwidth]{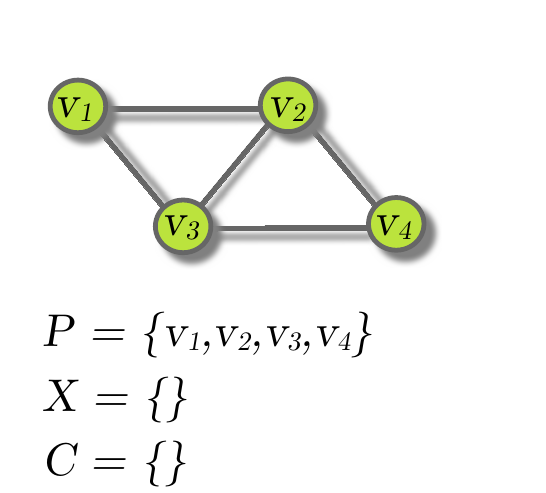}
    \subcaption{}\label{fig:tomita-a}
    \end{subfigure}
    \hfil
    \begin{subfigure}{0.24\textwidth}
    \centering
    \includegraphics[page=2, trim=100 260 430 115, clip, width=\textwidth]{figures/tomita-example-3.pdf}
    \subcaption{}\label{fig:tomita-b}
    \end{subfigure}
    \hfil
    \begin{subfigure}{0.24\textwidth}
    \centering
    \includegraphics[page=3, trim=100 260 430 115, clip, width=\textwidth]{figures/tomita-example-3.pdf}
    \subcaption{}\label{fig:tomita-c}
    \end{subfigure}
    \hfil
    \begin{subfigure}{0.24\textwidth}
    \centering
    \includegraphics[page=4, trim=100 260 430 115, clip, width=\textwidth]{figures/tomita-example-3.pdf}
    \subcaption{}\label{fig:tomita-d}
    \end{subfigure}
    \hfil
    \begin{subfigure}{0.24\textwidth}
    \centering
    \includegraphics[page=5, trim=100 260 430 115, clip, width=\textwidth]{figures/tomita-example-3.pdf}
    \subcaption{}\label{fig:tomita-e}
    \end{subfigure}
    \hfil
    \begin{subfigure}{0.24\textwidth}
    \centering
    \includegraphics[page=6, trim=100 260 430 115, clip, width=\textwidth]{figures/tomita-example-3.pdf}
    \subcaption{}\label{fig:tomita-f}
    \end{subfigure}
    \hfil
    \begin{subfigure}{0.24\textwidth}
    \centering
    \includegraphics[page=7, trim=100 260 430 115, clip, width=\textwidth]{figures/tomita-example-3.pdf}
    \subcaption{}\label{fig:tomita-g}
    \end{subfigure}
    \hfil
    \begin{subfigure}{0.24\textwidth}
    \centering
    \includegraphics[page=8, trim=100 260 430 115, clip, width=\textwidth]{figures/tomita-example-3.pdf}
    \subcaption{}\label{fig:tomita-h}
    \end{subfigure}
    \hfil
    \begin{subfigure}{0.24\textwidth}
    \centering
    \includegraphics[page=9, trim=100 260 430 115, clip, width=\textwidth]{figures/tomita-example-3.pdf}
    \subcaption{}\label{fig:tomita-i}
    \end{subfigure}
    \hfil
    \begin{subfigure}{0.24\textwidth}
    \centering
    \includegraphics[page=10, trim=100 260 430 115, clip, width=\textwidth]{figures/tomita-example-3.pdf}
    \subcaption{}\label{fig:tomita-j}
    \end{subfigure}
    \hfil
    \caption{A running example of \cref{alg:bk} on a small graph. Green, blue, and red vertices belong to $P$, $C$, and $X$, respectively. The current pivot vertex $v_p$ is circled with a dotted curve. The neighbourhood of $v_p$ is circled with a dashed curve, the current vertex $v$ is circled with a red solid curve.}
    \label{fig:tomita}
\end{figure}

A running example of \cref{alg:bk} on a small graph is depicted in \cref{fig:tomita}. When the first call of the procedure \texttt{ProcMCE} is launched, all the vertices are in $P$, while $X = C = \varnothing$ (\cref{fig:tomita-b}). The pivot vertex is $v_p = v_2$ (line 9 of \cref{alg:bk}) and the procedure iterates on vertex $v_2$ only, since all other vertices in $P$ are neighbors of $v_p$ (line 10 of \cref{alg:bk}). Hence, in the second call of \texttt{ProcMCE} (\cref{fig:tomita-c}) vertex $v_2$ has been moved to $C$, $v_p = v_3$, and line 10 of \cref{alg:bk} only iterates on $v_3$. Instead, Call~3 iterates on both $v_1$ and $v_4$: the iteration on $v_1$ (\cref{fig:tomita-d}) launches Call~4 (\cref{fig:tomita-e}) while the iteration on $v_4$ (\cref{fig:tomita-f}) launches Call~5 (\cref{fig:tomita-g}). Call~4 and Call~5 detect each a maximal clique, since both $P$ and $X$ are empty (line 6 of \cref{alg:bk}). When Call~5 terminates, the control returns to Call~3 (\cref{fig:tomita-h}) and then to Call~2 (\cref{fig:tomita-i}) and Call~1 (\cref{fig:tomita-j}), terminating the computation.

 \cref{alg:bk} requires 
$\bigO(3^{n/3})$ time~\cite{DBLP:journals/tcs/TomitaTT06}. While being optimal in the worst-case (matching the number of maximal cliques in Moon-Moser graphs \cite{Moon1965OnCI}), it has $\Omega(3^{n/6})$ delay and no polynomial-delay variant can exist, even with different pivoting, unless $\P = \NP$~\cite{DBLP:journals/tcs/ConteT22}. Despite these worst-case bounds, Bron-Kerbosch-based approaches are known to perform very efficiently in practice~\cite{DBLP:journals/jea/EppsteinLS13,DBLP:journals/tcs/TomitaTT06}.

\section{Pruning for $\ell$-isolated maximal clique enumeration}
\label{sec:enumeration}

In this section we show how to adapt \cref{alg:bk} to enumerate only $\ell$-isolated maximal cliques, introducing different pruning strategies to cut uninteresting  branches of the search tree. We also analyze the pruning capabilities and computational cost of each strategy, showing an inverse relation that suggests a natural order in which the heuristics may be applied in combination.

%We start with avg-isolation, since during the enumeration we have to rely on the size of the current maximal clique. Then, we consider min-isolation and max-isolation, since instead we only have to take into account the isolation factor $\ell$ during the enumeration.

\smallskip
\noindent{\bf Growing $\ell$-isolated cliques.} We recall from \cref{ss:mce} that in \cref{alg:bk}, during the enumeration, $C$ and $P$ are the current candidate clique and the perspective set, respectively. Let $T$ be any maximal clique in the subgraph $G[P]$ induced by $P$. Thanks to the invariant property that $P$ is always fully connected to $C$ during the execution, it is not difficult to see that $C \cup T$ is a clique in $G$, though not necessarily $\ell$-isolated. Our approach, when processing a node $r$ of the search tree $\mathcal{T}_G$ associated with a triple $(C,P,X)$, is to avoid recursive calls whenever the subtree rooted at $r$ cannot yield $\ell$-isolated maximal cliques.

First, we introduce a necessary condition for a clique grown from $C$ and $P$ to be $\ell$-isolated.
Let $T$ be any maximal clique in $G[P]$ and let $t$ be its size. By the invariant property maintained by \cref{alg:bk}, $C \cup T$ is a clique in $G$. Such a clique has at least $\delta_e(C,P) + |C|(|P| - t)$ external edges: each of the $|P| - t$ vertices not included in $C$ is indeed connected through an external edge with all the vertices in $C$. 
By definition, the number of external edges in an $\ell$-isolated clique must be strictly less than $\ell$ times the clique size. Hence, if
\begin{equation}
\label{eq:necessary}
    \delta_e(C,P) + |C|(|P| - t)\geq\ell(|C|+t)
\end{equation}
then $C \cup T$ is not $\ell$-isolated. We now introduce a preliminary result that will be useful to prove the correctness of our pruning strategies. 
%Let $G$, $\mathcal{T}_G$, and $\ell$ be the input graph, the search tree, and the isolation factor, respectively.

%For a given graph $G$, we consider in our analysis the search tree $\mathcal{T}_G$ of \cref{alg:bk} executed on graph $G$. Let $\ell$ be the isolation factor we want to achieve for the enumerated cliques.

\begin{lemma}\label{le:pruning-preliminary}
Let $r$ be a node in $\mathcal{T}_G$ associated to a candidate set $C$ and a perspective set $P$. Let $P_1$ and $P_2$ be any two subsets of $P$ such that $|P_1|\leq|P_2|$. 
If $\delta_e(C,P) + |C|(|P| - |P_2|)\geq\ell(|C|+|P_2|)$, then both $C\cup P_2$ and $C\cup P_1$ are not $\ell$-isolated.
\end{lemma}
\begin{proof}
% As a first step, we give a necessary condition for a clique grown from $C$ and $P$ to be $\ell$-isolated.
% Let $T$ be any maximal clique in $G[P]$ and let  $t$ be its size. By the invariant property maintained by \cref{alg:bk}, $C \cup T$ is a clique in $G$. Such a clique has at least $\delta_e(C) + |C|(|P| - t)$ external edges\todo{I:  maybe it's important to use $\delta_e(C,P)$ for this argument?}: each of the $|P| - t$ vertices not included in $C$ is indeed connected through an external edge with all the vertices in $C$. 
% By~\cref{def:l-isolation}, the number of external edges in an $\ell$-isolated clique must be strictly less than $\ell$ times the clique size. Hence,  if $\delta_e(C) + |C|(|P| - t)\geq\ell(|C|+t)$, then $C \cup T$ is not $\ell$-isolated.
%
%We now prove that no $\ell$-isolated clique can exist in the subtree rooted at $r$ when we prune.
%
%We now consider $P_1$ and $P_2$ as any two subsets of $P$ such that $C\cup P_1$ and $C\cup P_2$ are two maximal cliques\todo{I:  in realtà non serve che siano clique, ma l'isolatezza l'abbiamo definita solo su clique} and $|P_1|\leq|P_2|$. 
As observed when we introduced \cref{eq:necessary}, if $\delta_e(C,P) + |C|(|P| - |P_2|)\geq\ell(|C|+|P_2|)$, then $C\cup P_2$ is not $\ell$-isolated.
The relationship $|P_1|\leq |P_2|$ yields the following chain of inequalities: 
$\delta_e(C,P) + |C|(|P| - |P_1|) \geq \delta_e(C,P) + |C|(|P| - |P_2|) \geq \ell (|C| + |P_2|) \geq \ell (|C| + |P_1|)$.
Hence, $C\cup P_1$ is also not $\ell$-isolated.
%, proving the statement.
\end{proof}

\smallskip
\noindent{\bf Pruning by maximum cliques.} The size of the largest clique in the subgraph $G[P]$ induced by $P$, i.e., its clique number $\omega(G[P])$, can be used to prune subtrees of $\mathcal{T}_G$.
Namely, we prove in~\cref{le:pruning-omega} that it is correct to prune the subtree rooted at a node labelled $(C,P,X)$ whenever 
\begin{equation}
\label{eq:omega-equation}
    \delta_e(C,P) + |C| |P| - \ell |C| \geq \omega(G[P])\cdot(\ell + |C|)
\end{equation}
%\begin{equation}
%\label{eq:necessary}
%    \delta_e(C,P) + |C|(|P| - \omega(G[P]))\geq\ell(|C|+\omega(G[P]))
%\end{equation}

\begin{lemma}\label{le:pruning-omega}
Let $r$ be a node in $\mathcal{T}_G$ associated to a candidate set $C$ and a perspective set $P$. If $\delta_e(C,P) + |C| |P| - \ell |C| \geq \omega(G[P])\cdot(\ell + |C|)$, then the subtree of $\mathcal{T}_G$ rooted at $r$ contains no leaf corresponding to an $\ell$-isolated maximal clique.
\end{lemma}
\begin{proof}
We use~\cref{le:pruning-preliminary}, choosing $P_2$ as a maximum clique in $G[P]$, so that $|P_2| = \omega(G[P])$. Notice that the inequality $\delta_e(C,P) + |C||P| -\ell|C| \geq \omega(G[P])\cdot(\ell+|C|)$  in the statement can be rewritten as $\delta_e(C,P) + |C|\big(|P| - \omega(G[P])\big) \geq \ell \big(|C| + \omega(G[P])\big)$. When it is satisfied, by \cref{eq:necessary} the subgraph induced by $C$ extended by a maximum clique of $G[P]$ cannot be $\ell$-isolated. By~\cref{le:pruning-preliminary}, if we extend $C$ by any other clique of $G[P]$, the resulting subgraph is also not $\ell$-isolated. Any clique $P_1$ of $G[P]$ is indeed smaller than or equal to a maximum clique, i.e., $|P_1|\leq|P_2|=\omega(G[P])$ and thus~\cref{le:pruning-preliminary} can be applied.
\end{proof}

\smallskip
\noindent{\bf Pruning by upper bounding the maximum clique size.} \cref{le:pruning-omega} implies that a check on $\omega(G[P])$ would be sufficient to prune at node~$r$.  Unfortunately, computing $\omega(G[P])$ is difficult --- max-clique being \NP-hard. Hence, we exploit upper bounds to $\omega(G[P])$ to guide pruning, inspired by the approach suggested in~\cite{DBLP:conf/kdd/WangCF13} for the computation of $\tau$-visible summaries.
We will prove in~\cref{th:pruning-avg} that this maintains correctness with respect to our $\ell$-isolation problem.  
In particular, $\omega(G[P])$ cannot be larger than:  

%Then, let $d$ be the size of the maximum clique in $G[P]$. We upper\todo{I:  qui c'è molto da lavorare, anche sulla notazione. Vogliao usare $\omega(G[P])$ invece di d?} bound $d$ in several ways, using four heuristics introduced in~\cite{DBLP:conf/kdd/WangCF13}: 

\begin{enumerate}
    \item the number $|P|$ of vertices in $G[P]$;
    \item the maximum degree of $G[P]$, denoted by $\Delta(G[P])$, augmented by $1$;
    \item the maximum value of $k$ such that $G[P]$ contains at least $k$ vertices with degree $\geq k-1$, denoted by $\tau(G[P])$;
    \item the degeneracy of $G[P]$, denoted by $\kappa(G[P])$, augmented by $1$. The degeneracy is the largest value of $k$ for which $G[P]$ has a $k$-core (i.e., a maximal subgraph in which each vertex has degree at least $k$).
\end{enumerate}

%\begin{enumerate}
%    \item the size $d_p$ of $P$
%    \item the maximum degree $d_\Delta$ of $G[P]$
%    \item the maximum value $d_k$ of $k$ such that there are $k$ vertices with at least degree $k-1$ in $G[P]$
%    \item the degeneracy $d_d$ of $G[P]$
%\end{enumerate}

\noindent The first upper bound $\omega(G[P])\leq |P|$ is the weakest, while the fourth one $\omega(G[P])\leq \kappa(G[P])$ is the strictest. In particular, the following chain of inequalities holds:
\begin{equation}
\label{eq:upper-bounds}\omega(G[P])\leq\kappa(G[P]) + 1\leq \tau(G[P])\leq \Delta(G[P])+1 \leq |P|\end{equation}
This can be intuitively explained as follows. The fact that $\omega(G[P])\leq \kappa(G[P]) + 1$ is well-known in the literature: in every clique of size $t$, every vertex has indeed degree $t-1$ and is thus part of a $(t-1)$-core, proving that the degeneracy is at least $t-1$. The value $\tau(G[P])$ represents a relaxed version of degeneracy, and thus $\kappa(G[P]) + 1\leq \tau(G[P])$. Indeed, the degeneracy $\kappa(G[P])$ can be computed by repeatedly removing a vertex of minimum degree from the graph, and updating the degree of remaining vertices along the way. Its relaxation $\tau(G[P])$ can be instead obtained by removing vertices based on their original degrees, without updates. Hence, it may be the case that more vertices satisfy the degree condition at point 3. Consider, for instance, a complete binary tree on four levels: $\tau(G[P])=4$, as the maximum degree 3 is attained by all the $6$ vertices in the two internal layers, while $\kappa(G[P]) + 1=2$, as the degeneracy of a tree is~$1$. Finally, it can be easily seen that $\tau(G[P])\leq \Delta(G[P])+1$, because we cannot have $k$ vertices of degree larger than or equal to $k-1$ if $k\geq \Delta(G[P])+2$.

We use the upper bounds on $\omega(G[P])$ to obtain lower bounds on the total number of external edges that any maximal clique $T$, grown from the current clique $C$ and the candidate set $P$, can have. Namely, our algorithm prunes the computation at node $(C,P,X)$ whenever 
\begin{equation}
\label{eq:upper-bound-pruning}
    \delta_e(C,P) + |C| |P| - \ell |C| \geq \overline{\omega}(G[P])\cdot(\ell + |C|)
\end{equation}
where $\overline{\omega}(G[P])$ can be any of the quantities involved in~\cref{eq:upper-bounds} (i.e., either $\omega(G[P])$ or any of its upper bounds). 
In the following we discuss the correctness of the pruning condition expressed by~\cref{eq:upper-bound-pruning}, proving that it is safe to prune using upper bounds $\overline{\omega}$ on $\omega$, instead of $\omega$ itself. 

%Observe\todo{I:  ??? Sono persa su questa frase. A che serve?} that during the enumeration process, the pivot of $T$ could belong to $P$ (and not to $C$).

\begin{theorem}\label{th:pruning-avg}
    Given a graph $G$ and a positive integer $\ell$, let $\mathcal{T}_G$ be the search tree of \cref{alg:bk} executed on $G$, let $r$ be a node in $\mathcal{T}_G$ associated to a candidate set $C$ and a perspective set $P$, and let $\overline{\omega}(G[P])$ be any upper bound on $\omega(G[P])$. If $\delta_e(C,P) + |C| |P| - \ell |C| \geq \overline{\omega}(G[P])\cdot(\ell + |C|)$, then the subtree of $\mathcal{T}_G$ rooted at $r$ contains no leaf corresponding to an $\ell$-isolated maximal clique.
\end{theorem}

\begin{proof}
We exploit~\cref{le:pruning-preliminary}, choosing $P_1$ as a maximum clique in $G[P]$ and choosing $P_2$ as follows:
%\todo{I:  minor issue. $C\cup P_2$ may not be a clique. It's not a problem, but technically the argument could not be applied as stated.} 
%corresponding to the indexes involved in~\cref{eq:upper-bounds}. 
if the upper bound $\overline{\omega}(G[P])$ in~\cref{eq:upper-bounds} is $|P|$, then $P_2$ is $P$ itself; if instead it is $\Delta(G[P]) + 1$, then $P_2$ is formed by any vertex of maximum degree in $P$ together with its neighbors; if it is $\tau(G[P])$, then $P_2$ is formed by any $\tau(G[P])$ vertices of $P$ each having degree at least $\tau(G[P]) - 1$; and when it is $\kappa(G[P]) + 1$, then $P_2$ is the $\kappa(G[P])$-core of $P$.
Hence, $|P_1|=\omega(G[P])\leq \overline\omega(G[P])=|P_2|$. If the test in \cref{eq:upper-bound-pruning} is verified, by \cref{le:pruning-preliminary} we have that the subgraph induced by $C$ extended by a maximum clique of $G[P]$ is not $\ell$-isolated. In particular, from the proof of \cref{le:pruning-preliminary} we have $\delta_e(C,P) + |C|(|P| - \omega(G[P])) \geq \ell (|C| + \omega(G[P])$ (\cref{eq:omega-equation}). 
This in turn implies, by~\cref{le:pruning-omega}, that any other clique in $G[P]$ is not $\ell$-isolated and that pruning at $r$ does not entail discarding any $\ell$-isolated clique, thereby proving the statement.
\end{proof}

We conclude with a remark on the running time: the quantities $|P|$, $\Delta(G[P])$, $\tau(G[P])$, and $\kappa(G[P])$ involved in~\cref{eq:upper-bounds} can be computed in $O(1)$, $O(1)$, $O(|P|)$, and $O(|E(G[P])|)$ time, respectively, yielding an inverse relation between running time and pruning power.

\section{Experimental setup}\label{sec:setup}

\smallskip
\noindent{\bf Algorithms under evaluation.}
The four different heuristics discussed in \cref{sec:enumeration} are denoted in our analysis by \size, \degree, \softcore, and \degeneracy, respectively. We also consider a combined approach, called \combo, that orderly integrates the \size\ and \softcore\ heuristics: we tested many different combinations of the four basic heuristics, and \combo\ consistently proved to be the most effective one.
For the sake of comparison, we include the algorithm proposed in~\cite{DBLP:journals/tcs/HuffnerKMN09}, denoted by \hkmn.
We relied on the implementation used in~\cite{DBLP:journals/tcs/HuffnerKMN09}.
% ~\cite{DBLP:journals/talg/ItoI09,DBLP:conf/esa/ItoIO05}.}
Finally, we consider a baseline algorithm, derived from \cref{alg:bk}~\cite{DBLP:conf/cocoon/TomitaTT04,DBLP:journals/tcs/TomitaTT06}
% ~\cite{DBLP:journals/tcs/TomitaTT06}
and denoted by \ttt, 
%~\cite{DBLP:conf/cocoon/TomitaTT04,DBLP:journals/tcs/TomitaTT06}, %in which only $\ell$-isolated maximal cliques are reported 
that filters each non-$\ell$-isolated maximal clique instead of reporting it. 

\begin{table}[t]
  \caption{Statistics of the real-world benchmarks. $\Delta$ is the maximum degree, and $d$ is the degeneracy.\vspace{2mm}}
  \centering
  \resizebox{\textwidth}
  %\vspace{5mm}
  {!}{

{%\footnotesize

\begin{tabular}{l|l|l|r|r|r|r|r|r}
    \toprule
    \textbf{Graph} & \textbf{ID} & \textbf{Category} & \textbf{Vertices} & \textbf{Edges} & \textbf{Max. Cliques} & $\Delta$ & \textbf{d} \\
    \midrule
    brightkite & bk           & social & 58,228    & 214,078   & 290,004    & 1,134   & 52  \\
    livemocha  & lm           & social  & 104,103   & 2,193,083 & 3,711,569  & 2,980   & 92 \\
    gowalla  & gw           & social & 196,591   & 950,327   & 1,212,679  & 14,730  & 51 \\
    twitter  & ws           & social  &  465,017  &  833,540  & 805,170   & 677   & 30  \\
    youtube   & cy           & social & 1,134,890	& 2,987,624   & 3,265,956   &  28,754  & 51  \\
    hyves    & hy            & social  & 1,402,673 & 2,777,419 & 2,590,217  & 31,883  & 39  \\
   
    \midrule
    ca-astroph  & ap          & co-authorship    & 18,771    & 198,050   & 36,427     & 504     & 56  \\
    dblp      & cd           & co-authorship    & 317,080   & 1,049,866 & 257,551    & 343     & 113 \\
    \midrule
%    political-blogs (NO)      & hyperlink        & 1,224     & 16,715    & 49,618     & 351     & 36  & 1.414 \\
%    wikipedia-link-nah   & hyperlink        & 10,285    & 142,711   & 87,873     & 2,698   & 99  & 1.453 \\
    cfinder-google  & gc     & hyperlink        & 15,763    & 148,585   & 75,258     & 11,401  & 102 \\
    stanford   & sf         & hyperlink        & 281,903   & 1,992,636 & 1,055,936  & 38,625  & 71  \\
    italian-cnr  & ic       & hyperlink        & 325,557   & 2,738,969 & 1,425,378  & 18,236  & 83  \\
    notre-dame  & nd        & hyperlink        & 325,729   & 1,090,108 & 495,947    & 10,721  & 155 \\
    baidu     & bar           & hyperlink        & 415,624   & 2,374,044 & 4,596,143  & 127,066 & 228 \\
    web-google  & go         & hyperlink        & 875,713   & 4,322,051 & 1,417,580  & 6,332   & 44  \\
    \midrule
    cit-hepph  & phc           & citation & 34,546  & 420,877 & 412,493  & 846   & 30  \\
    citeseer   & cs          & citation & 384,054   & 1,736,145 & 1,232,951  & 1,739   & 15  \\
    \bottomrule
    \end{tabular}
}
  }
  \label{tab:graph-stats}
\end{table}

\smallskip
\noindent{\bf{Benchmarks.}} 
We conducted experiments on a variety of real-world and synthetic networks.
Real-world networks, which are the main target of our analysis, are taken primarily from the KONECT repository ({\url{http://konect.cc}})~\cite{DBLP:conf/www/Kunegis13}. Throughout this paper we focus on a representative subset of $16$ graphs, whose main properties are summarized in~\cref{tab:graph-stats}. 
Our selection was guided by diversity criteria, taking into account {\em graph size}, {\em number of maximal cliques}, {\em maximum degree} $\Delta$, and {\em degeneracy} $d$. 
To further ensure heterogeneity, our benchmarks include different types of networks, namely {social}, {hyperlink}, {co-authorship}, and {citation} graphs.  
For the sake of completeness, the number of $\ell$-isolated maximal cliques for different values of $\ell$, as well as their percentage with respect to the total number of maximal cliques, is reported for these graphs in \cref{tab:number-of-maximal-cliques-varying-l}.

\begin{table}[htbp]
  \centering
  \caption{Evolution of the number of $\ell$-isolated maximal cliques and the percentage with respect to all maximal cliques for different values of $\ell$.\vspace{2mm}}
  \label{tab:number-of-maximal-cliques-varying-l}
  \resizebox{\textwidth}{!}{
    \begin{tabular}{l|rrrrrrrrrr}
    \toprule
    \textbf{Graph} & \multicolumn{1}{c}{$\ell=1$} & \multicolumn{1}{c}{$\ell=10$} & \multicolumn{1}{c}{$\ell=20$} & \multicolumn{1}{c}{$\ell=30$} & \multicolumn{1}{c}{$\ell=40$} & \multicolumn{1}{c}{$\ell=50$} & \multicolumn{1}{c}{$\ell=100$} & \multicolumn{1}{c}{$\ell=150$} & \multicolumn{1}{c}{$\ell=200$} & \multicolumn{1}{c}{$\ell=250$} \\
    \midrule
brightkite & 2346 & 32661 & 54118 & 67631 & 77577 & 85257 & 113328 & 181361 & 232791 & 264373 \\
           & (0.81\%) & (11.26\%) & (18.66\%) & (23.32\%) & (26.75\%) & (29.40\%) & (39.08\%) & (62.54\%) & (80.27\%) & (91.16\%) \\
\midrule
livemocha & 2 & 6619 & 32211 & 69514 & 112054 & 156226 & 361422 & 528104 & 668495 & 792579 \\
          & (0.00\%) & (0.18\%) & (0.87\%) & (1.87\%) & (3.02\%) & (4.21\%) & (9.74\%) & (14.23\%) & (18.01\%) & (21.35\%) \\
\midrule
gowalla & 3993 & 107537 & 176808 & 220693 & 252900 & 279548 & 376567 & 440791 & 511168 & 565420 \\
        & (0.33\%) & (8.87\%) & (14.58\%) & (18.20\%) & (20.85\%) & (23.05\%) & (31.05\%) & (36.35\%) & (42.15\%) & (46.63\%) \\
\midrule
twitter & 8 & 803 & 3401 & 6678 & 10471 & 15811 & 47327 & 82293 & 124235 & 603666 \\
        & (0.00\%) & (0.10\%) & (0.42\%) & (0.83\%) & (1.30\%) & (1.96\%) & (5.88\%) & (10.22\%) & (15.43\%) & (74.97\%) \\
\midrule
youtube & 41563 & 549757 & 756313 & 882061 & 975738 & 1049952 & 1291485 & 1450548 & 1566338 & 1662987 \\
        & (1.27\%) & (16.83\%) & (23.16\%) & (27.01\%) & (29.88\%) & (32.15\%) & (39.54\%) & (44.41\%) & (47.96\%) & (50.92\%) \\
\midrule
hyves & 44 & 829384 & 1218221 & 1295354 & 1394493 & 1515859 & 1920881 & 2018340 & 2113061 & 2190694 \\
      & (0.00\%) & (32.02\%) & (47.03\%) & (50.01\%) & (53.84\%) & (58.52\%) & (74.16\%) & (77.92\%) & (81.58\%) & (84.58\%) \\
\midrule
ca-astroph & 384 & 2629 & 5157 & 7462 & 9732 & 11941 & 21594 & 28811 & 32935 & 35077 \\
           & (1.05\%) & (7.22\%) & (14.16\%) & (20.48\%) & (26.72\%) & (32.78\%) & (59.28\%) & (79.09\%) & (90.41\%) & (96.29\%) \\
\midrule
dblp & 4762 & 103815 & 168752 & 205317 & 226139 & 237806 & 255468 & 257404 & 257537 & 257550 \\
     & (1.85\%) & (40.31\%) & (65.52\%) & (79.72\%) & (87.80\%) & (92.33\%) & (99.19\%) & (99.94\%) & (99.99\%) & (100.00\%) \\
\midrule
cfinder-google & 92 & 657 & 1097 & 1471 & 1597 & 1833 & 3679 & 5623 & 6858 & 17418 \\
               & (0.12\%) & (0.87\%) & (1.46\%) & (1.95\%) & (2.12\%) & (2.44\%) & (4.89\%) & (7.47\%) & (9.11\%) & (23.14\%) \\
\midrule
stanford & 1223 & 65655 & 107570 & 133984 & 153565 & 170323 & 217294 & 234531 & 249591 & 262519 \\
         & (0.12\%) & (6.22\%) & (10.19\%) & (12.69\%) & (14.54\%) & (16.13\%) & (20.58\%) & (22.21\%) & (23.64\%) & (24.86\%) \\
\midrule
italian-cnr & 1455 & 89276 & 141082 & 172788 & 193906 & 210687 & 267492 & 320474 & 369278 & 399686 \\
            & (0.10\%) & (6.26\%) & (9.90\%) & (12.12\%) & (13.60\%) & (14.78\%) & (18.77\%) & (22.48\%) & (25.91\%) & (28.04\%) \\
\midrule
notre-dame & 1292 & 108550 & 175070 & 211801 & 234167 & 251205 & 290741 & 317485 & 332918 & 345365 \\
           & (0.26\%) & (21.89\%) & (35.30\%) & (42.71\%) & (47.22\%) & (50.65\%) & (58.62\%) & (64.02\%) & (67.13\%) & (69.64\%) \\
\midrule
baidu & 22792 & 173799 & 306856 & 406504 & 466079 & 519940 & 631133 & 1031465 & 3280483 & 3301681 \\
      & (0.50\%) & (3.78\%) & (6.68\%) & (8.84\%) & (10.14\%) & (11.31\%) & (13.73\%) & (22.44\%) & (71.37\%) & (71.84\%) \\
\midrule
web-google & 8745 & 324503 & 535311 & 643214 & 709705 & 757338 & 888481 & 944254 & 975307 & 1003754 \\
           & (0.62\%) & (22.89\%) & (37.76\%) & (45.37\%) & (50.06\%) & (53.42\%) & (62.68\%) & (66.61\%) & (68.80\%) & (70.81\%) \\
\midrule
cit-hepph & 141 & 7273 & 23354 & 45525 & 71242 & 98142 & 251893 & 355771 & 390060 & 402261 \\
          & (0.03\%) & (1.76\%) & (5.66\%) & (11.04\%) & (17.27\%) & (23.79\%) & (61.07\%) & (86.25\%) & (94.56\%) & (97.52\%) \\
\midrule
citeseer & 10965 & 247124 & 478583 & 635481 & 746623 & 829906 & 1043603 & 1126876 & 1164634 & 1187575 \\
         & (0.89\%) & (20.04\%) & (38.82\%) & (51.54\%) & (60.56\%) & (67.31\%) & (84.64\%) & (91.40\%) & (94.46\%) & (96.32\%) \\
    \bottomrule
    \end{tabular}
  }
\end{table}

We also used synthetic instances generated according to two different models:

\begin{itemize}
\item Random scale-free networks generated using a preferential attachment mechanism as proposed by Barabasi and Albert~\cite{BA02}. $\textrm{BA}_{n,m}$ networks are grown by sequentially adding new nodes, each with $m$ edges, that attach preferentially to existing nodes in proportion to their degrees. We used $n = 100,000$ and $m \in \{25 , 50\}$.

\item Random networks based on the $G_{n,m,p}$ model (see~\cite{DBLP:journals/tcs/BehrischT06} and the references therein), which is also used in~\cite{DBLP:journals/tcs/HuffnerKMN09}. These networks have $n$ vertices; each of them is assigned $m$ {\em features}, each with probability $p$; and all the vertices with the same feature form a clique.
On these instances, for the sake of reproducibility, we selected parameter values similar to~\cite{DBLP:journals/tcs/HuffnerKMN09}, varying
%\todo{I:  forse i parametri specifici possiamo toglierli, per simmetria con BA in cui non ci conviene discuterli} 
$n$ in $[50,550]$ with step $50$, $m$ in $[5,95]$ with step $10$, $p$ in $[0.025,0.25]$ with step $0.025$, and $\ell$ in $[10,200]$ with step $10$.
\end{itemize}

\smallskip
\noindent{\bf{Performance metrics.}}
In addition to running times, we profiled the number of recursive calls and contrasted these values with the number of $\ell$-isolated cliques. The running time of each heuristic depends indeed both on its pruning capabilities and on the overhead introduced by the pruning test in~\cref{eq:upper-bound-pruning}: the more complex the pruning strategy is, the larger the test overhead.

\smallskip
\noindent{\bf{Experimental platform.}}
All the experiments were run on a laptop equipped with an Intel\textregistered{} Core\texttrademark{} i5-1135G7 CPU (4 cores, 4.2 GHz) and 8 GB of RAM. Our algorithms were implemented in OCaml (\url{https://ocaml.org}) to ensure a fair comparison with the available implementation used in~\cite{DBLP:journals/tcs/HuffnerKMN09}.

\section{Experimental results}\label{sec:experiments}

In this section we present the key findings of our experimental analysis. \cref{ss:distribution} examines the properties of $\ell$-isolated cliques in real-world benchmarks. 
We evaluate our pruning strategies in \cref{ss:pruning}, identifying the most effective ones, which are  compared against state-of-the-art methods in \cref{ss:competitors}. We complement our study with results on synthetic instances in \cref{app:experiments-synthetic}.

\subsection{On the distribution of $\ell$-isolated maximal cliques}
\label{ss:distribution}

\begin{figure}[t]
\centering\includegraphics[width=1\textwidth]{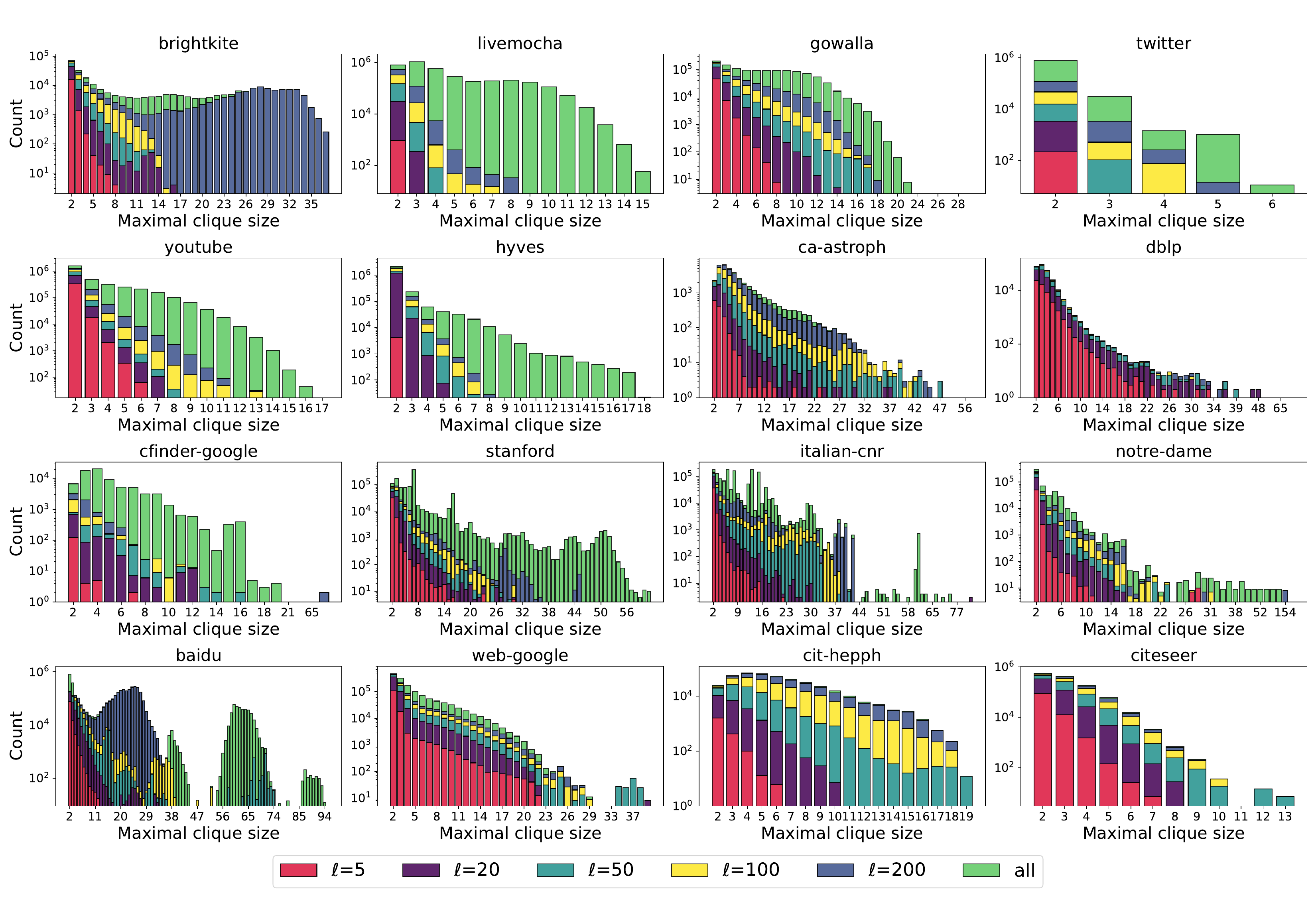}
\vspace{-6mm}
\caption{$\ell$-isolated clique size distribution for different values of $\ell$.}
\label{fig:clique-distribution-real}
\end{figure}

In \cref{fig:clique-distribution-real} we summarize the distribution of maximal cliques by size, together with their breakdown into $\ell$-isolated cliques for different values of $\ell$. 

The upper profile of the stacked bars shows the total number of maximal cliques with a given number of vertices. In all benchmarks, most cliques consist of a few vertices: the number of cliques grows quickly for small sizes, then drops sharply. The maximum clique size varies substantially across graph types (cfr. \cref{tab:graph-stats}), reflecting different connectivity patterns: 
the first six social networks
%\todo{I:  va riordinata la figura mettendo cit-hepph e citeseer in mezzo (prima di astroph). Forse anche le tabelle e altre figure?} 
form only small cliques (typically less than 20 vertices), the two co-authorship networks and the two citation networks reach moderate sizes, while the six Web graphs contain the largest cliques, sometimes exceeding 150 vertices.
%Social networks have mostly small cliques, very few beyond 20 vertices: social ties are local and sparse and friends rarely form very large fully-connected groups. Scientific collaboration and citation networks have moderate-to-large cliques (20–60 vertices), highlighting larger teams of researchers working on related topics. Web graphs exhibit rather large maximum cliques, with over 150 vertices in some cases, probably due to structural artifacts such as link farms. 
The difference can be best appreciated on graphs of different types but similar size: compare, e.g., the maximum clique size in {\small\tt gowalla}, {\small\tt dblp}, and {\small\tt notre-dame}, which have comparable numbers of edges. 

This analysis can be refined by considering the color segments within each bar of \cref{fig:clique-distribution-real}, that show how many maximal cliques are $\ell$-isolated for several values of $\ell$. 
%provide a decomposition of the total number of maximal cliques into subgroups depending on how isolated they are: the larger is $\ell$, the looser are the connections allowed to the outside.
%The color segments within each bar refine this view by decomposing the total number of maximal cliques into subgroups depending on how isolated they are: the larger is $\ell$, the looser are the connections allowed to the outside.
%Intuitively, an $\ell$-isolated clique can be seen as a relatively ``self-contained'' community, and the larger is $\ell$, the looser are the connections allowed to the outside.
For small values of $\ell$ (e.g., $\ell\leq 20$), the isolation requirement is very strict and only a small fraction of small cliques qualify: large cliques are typically excluded, as they tend to be embedded in denser regions of the graph with many external connections. The distributions for small $\ell$ are thus sharply truncated and dominated by cliques with just a few vertices. As $\ell$ increases, the number of $\ell$-isolated cliques grows, being the property monotone (see \cref{ss:l-isolation}).
%(see \cref{obs:monotone}). 
In addition, larger values of $\ell$ allow increasingly larger cliques to appear in the summary, extending the tail of the distribution. For large thresholds, such as $\ell=200$, many more cliques are admitted: in many networks (e.g., {\small\tt brightkite}, {\small\tt web-google}, {\small\tt cit-hepph}) the profile for 200-isolated cliques approaches that of the total distribution, making the isolation criterion almost ineffective.

\subsection{Assessment of different pruning strategies}
\label{ss:pruning}

\begin{figure}[t]
    \centering
\includegraphics[width=0.49\textwidth]{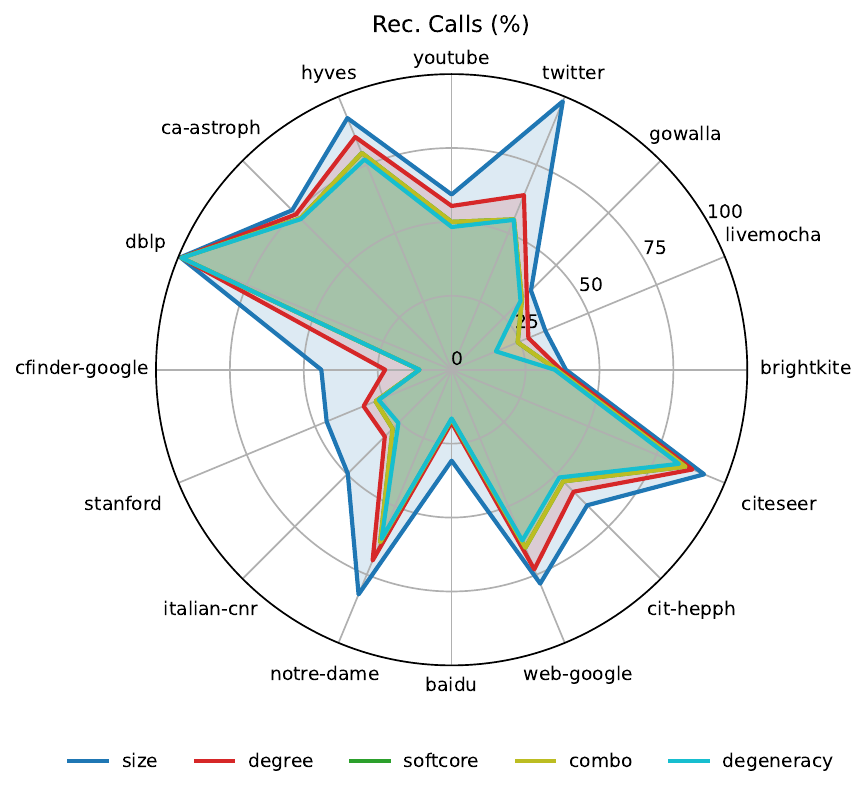}
\includegraphics[width=0.49\textwidth]{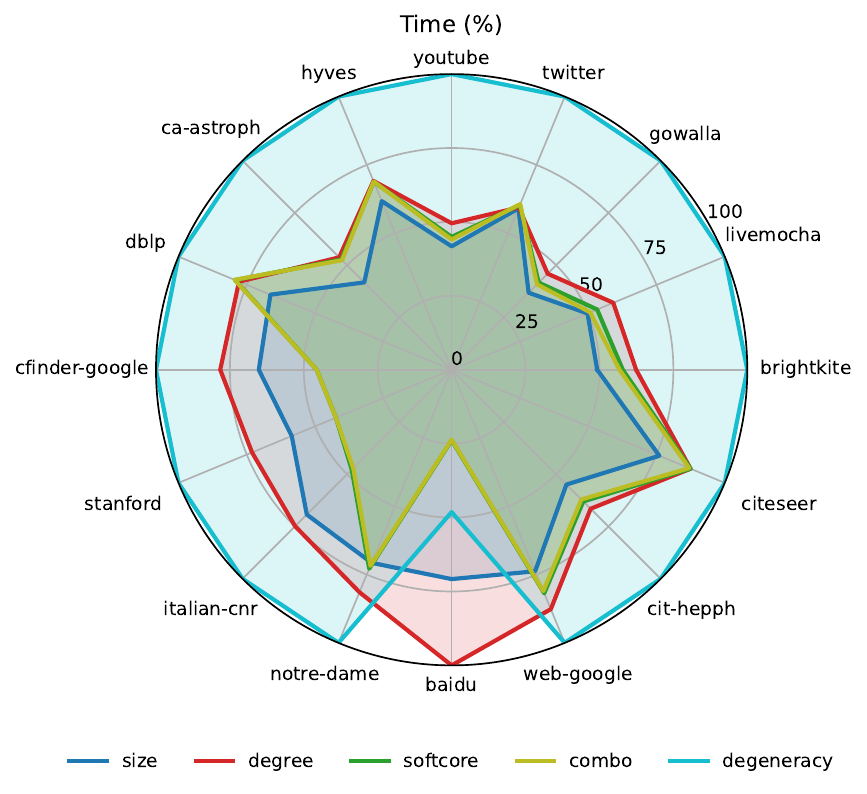}
\vspace{-1mm}
\caption{Comparison of the recursive calls (\%) and the running time (\% with respect to the slowest execution) of different pruning strategies for $\ell=50$.}
\label{fig:comparison-heuristics2}
\end{figure}

We now assess the effectiveness of different pruning strategies considering two different metrics: number of recursive calls and running time. 
%\cref{fig:comparison-heuristics2} only reports the results of the experiments made with $\ell=50$, being the results substantially the same for other values of $\ell$. 
Because the results are largely consistent across different values of $\ell$, \cref{fig:comparison-heuristics2} shows only the case $\ell=50$.
The left radar plot in \cref{fig:comparison-heuristics2} reports the percentage of recursive calls with respect to the number one would have without pruning. This gives a quantitative view of the pruning power of the heuristics: the larger the area, the weaker the pruning. In spite of a large variability across benchmarks, \size\ is always the weakest, as expected from the theoretical analysis, producing the largest number of calls. On some benchmarks, e.g., {\small\tt twitter}, this number can be much larger than the other strategies and very close to the baseline. As expected, \combo\ and \softcore\ overlap perfectly, since \combo\ is \size\ followed by \softcore\ whenever \size\ fails, thus inheriting the pruning power of \softcore. While \degeneracy\ consistently achieves the strongest pruning, yielding the smallest areas across all datasets, it is rather similar to \softcore\ and \combo.

% \begin{figure}[htbp]
%     \centering
% \includegraphics[width=1\textwidth]{figures/heur_recursive_50.pdf}
% \includegraphics[width=1\textwidth]{figures/heur_50.pdf}
% \vspace{-10mm}
% \caption{Comparison of the recursive calls (\%) and the running time of different pruning strategies on a selection of benchmarks for $\ell=50$.}
% \label{fig:comparison-heuristics}
% \end{figure}

\begin{figure}[thp]
    \centering
    % Riga 1
    \includegraphics[width=0.32\textwidth]{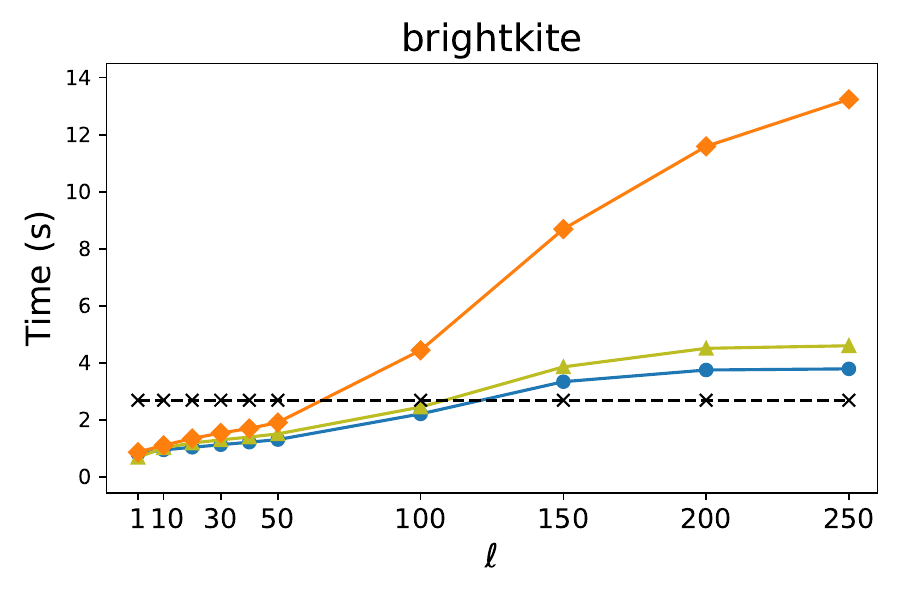}
    \includegraphics[width=0.32\textwidth]{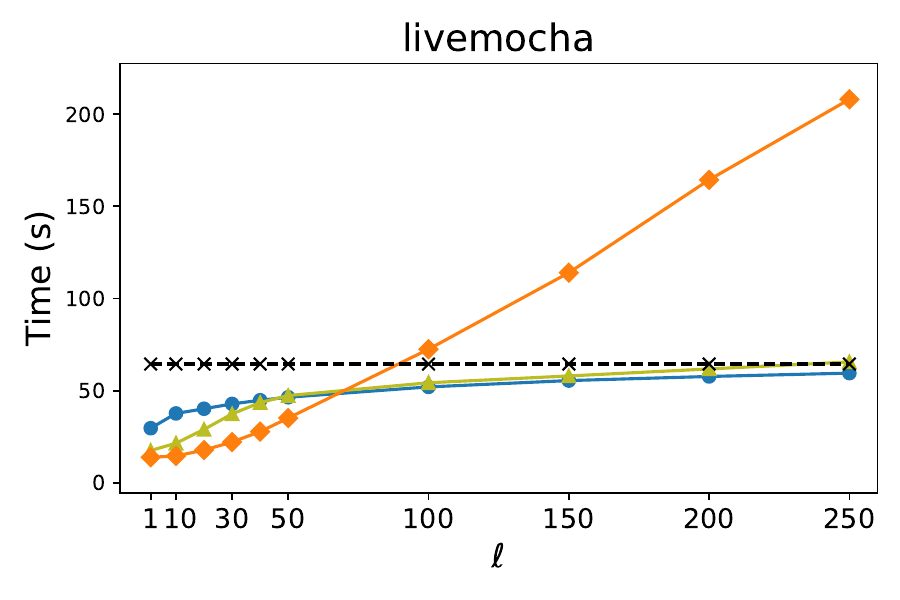}
    \includegraphics[width=0.32\textwidth]{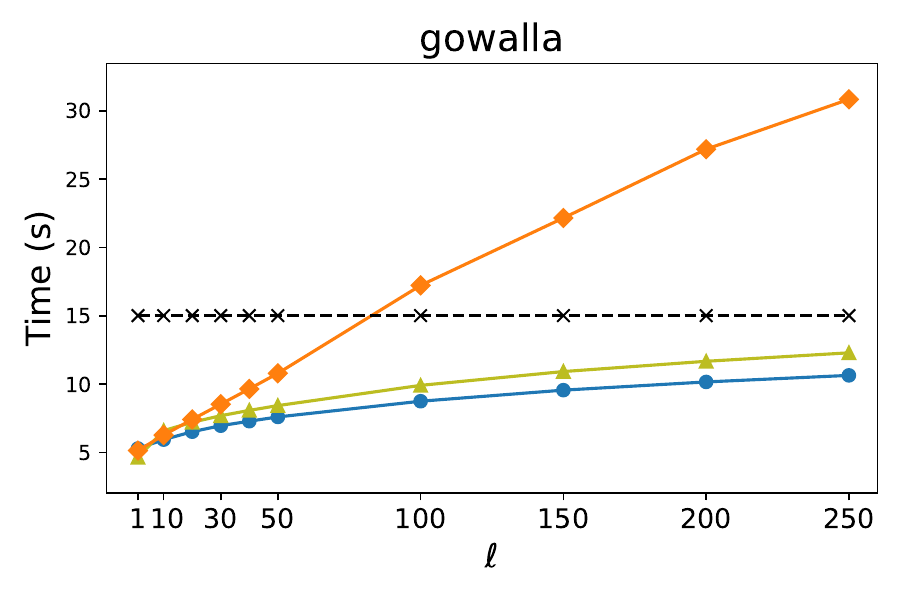}\\
    % Riga 2
    \includegraphics[width=0.32\textwidth]{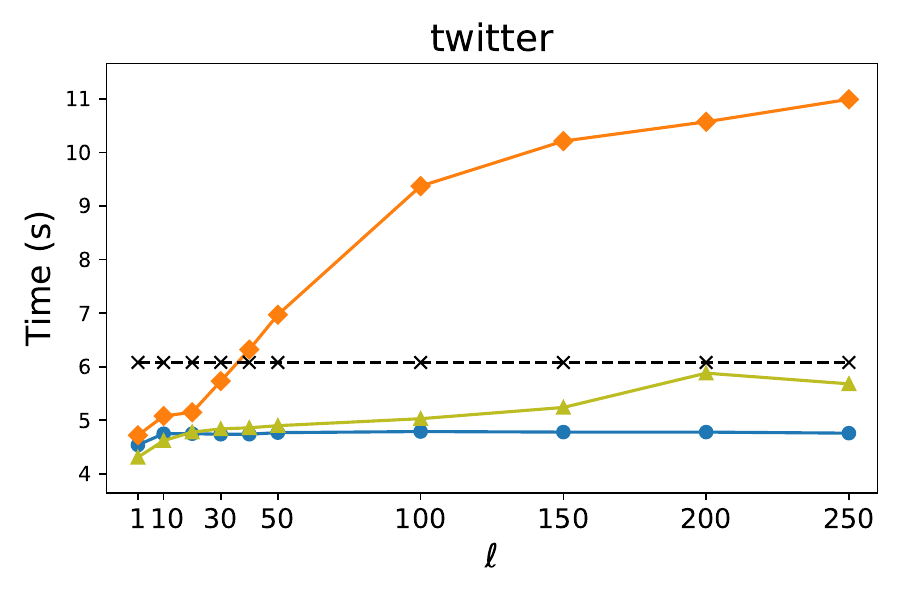}
    \includegraphics[width=0.32\textwidth]{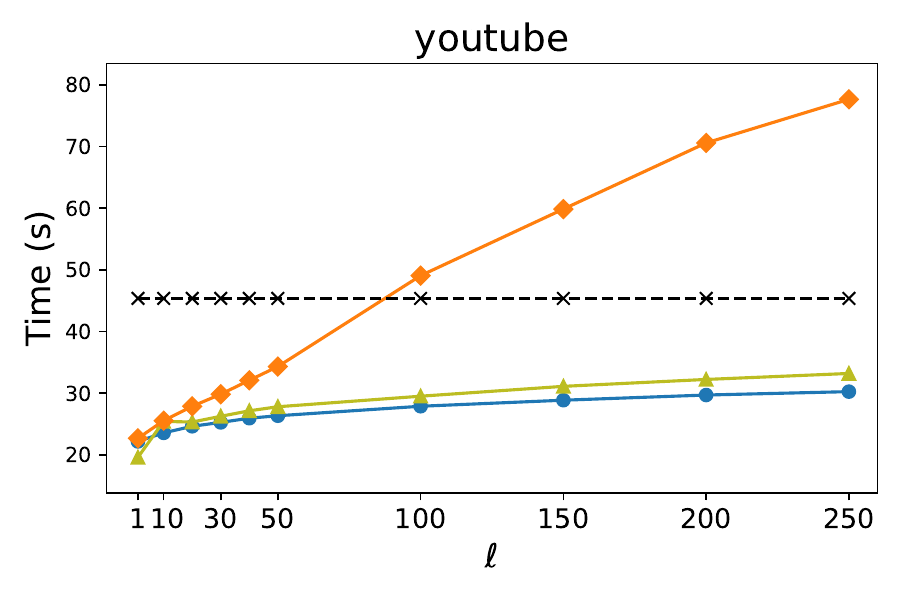}
    \includegraphics[width=0.32\textwidth]{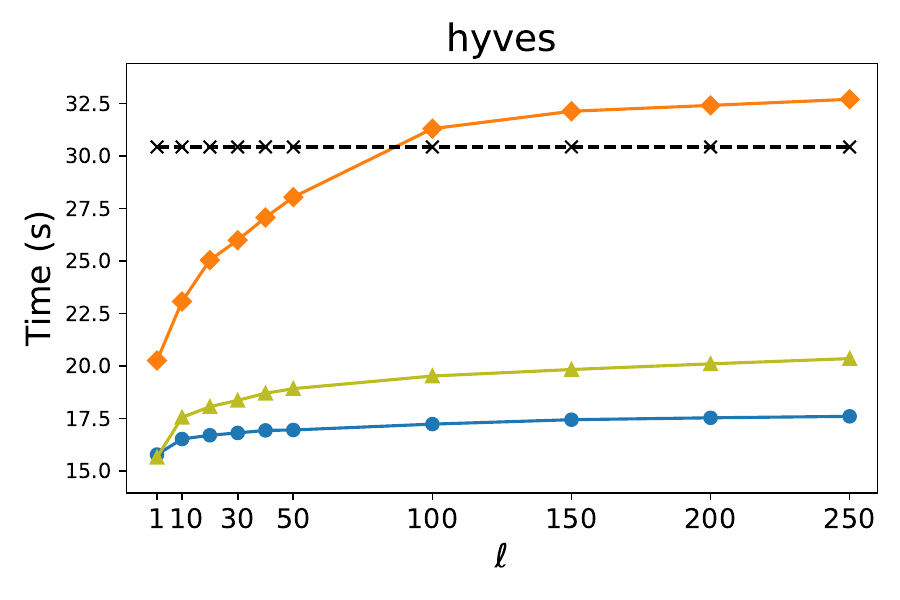}\\
    % Riga 3
    \includegraphics[width=0.32\textwidth]{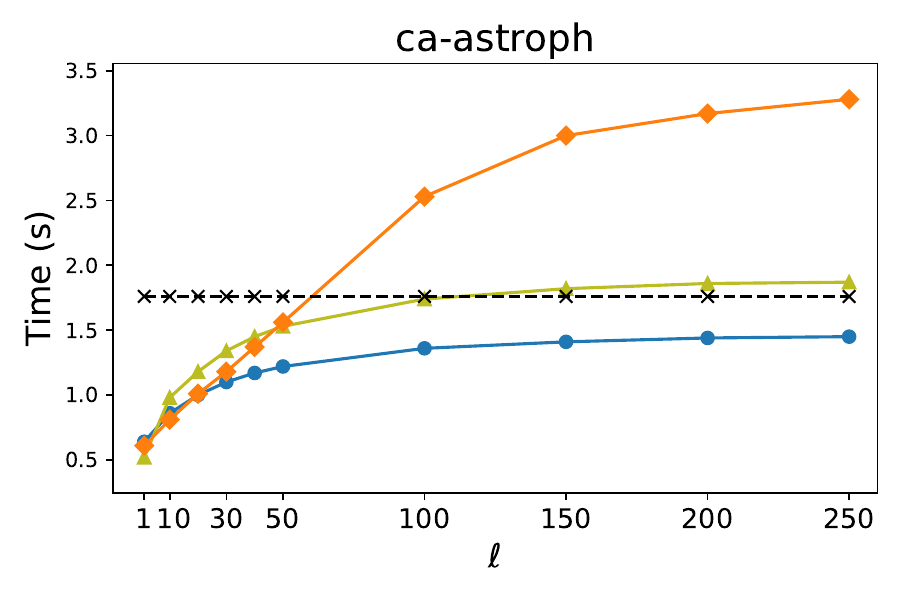}
    \includegraphics[width=0.32\textwidth]{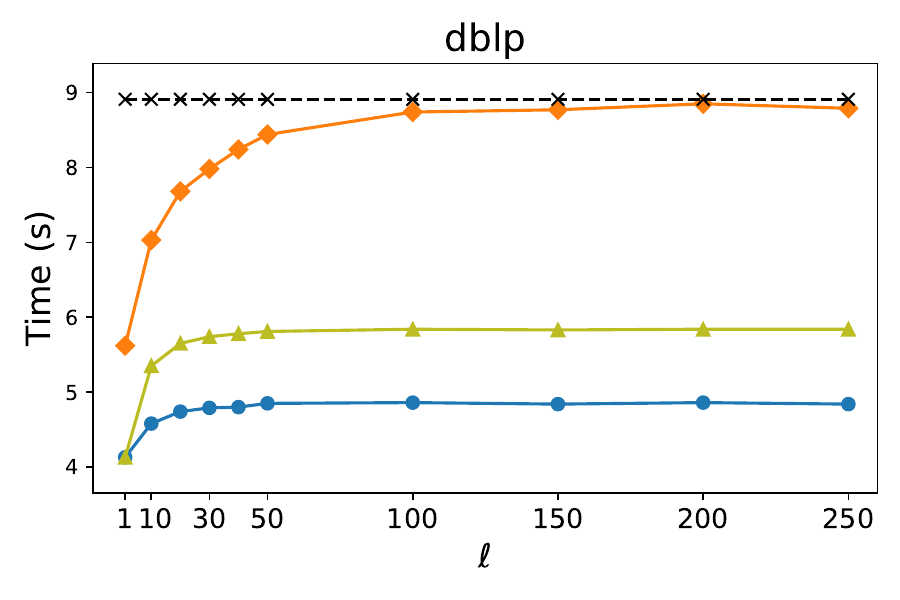} \\
    % Riga 4
    \includegraphics[width=0.32\textwidth]{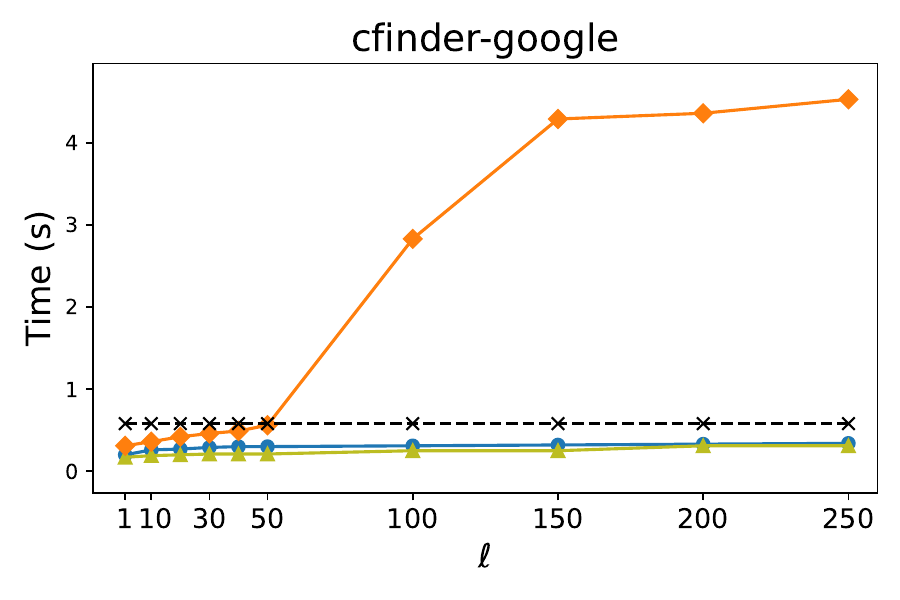}
    \includegraphics[width=0.32\textwidth]{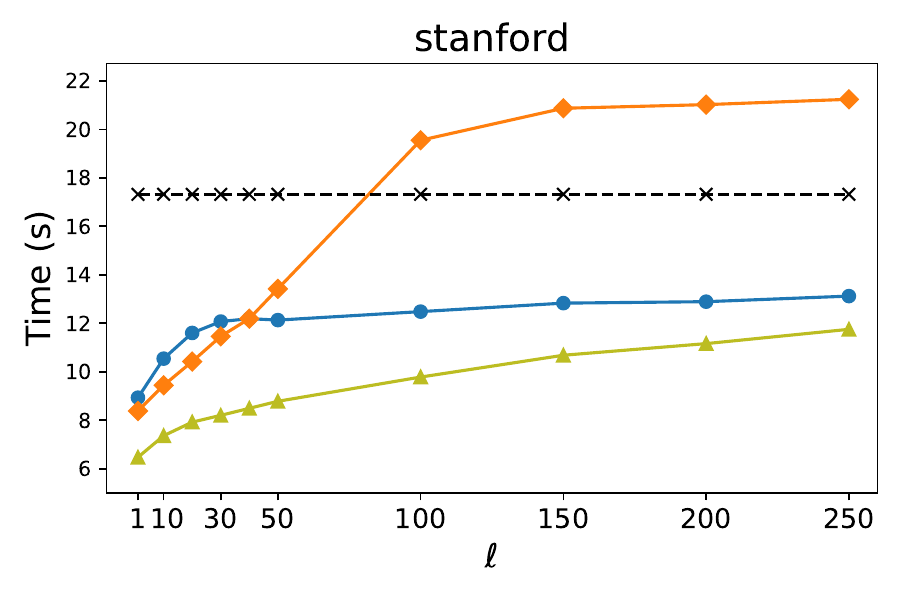}
    \includegraphics[width=0.32\textwidth]{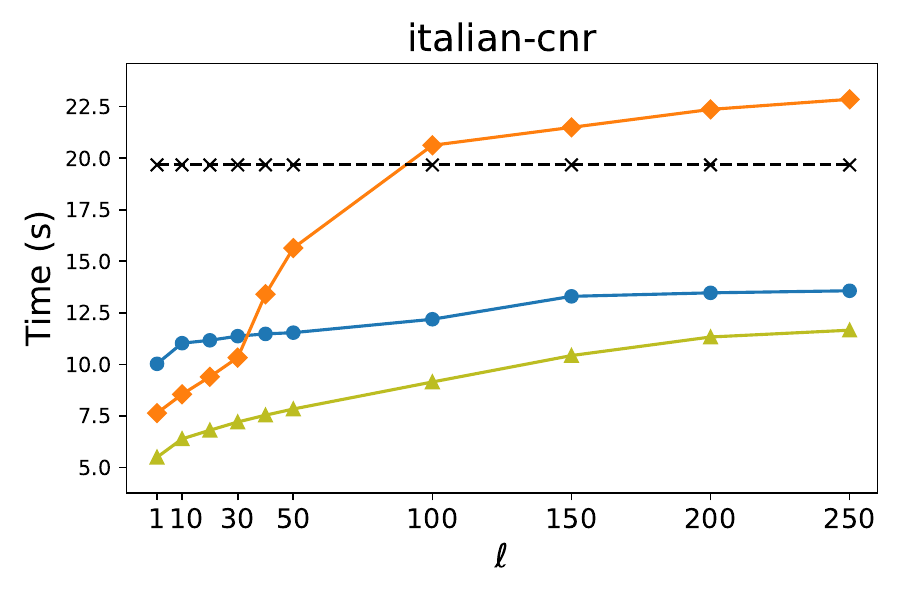} \\
    % Riga 5
    \includegraphics[width=0.32\textwidth]{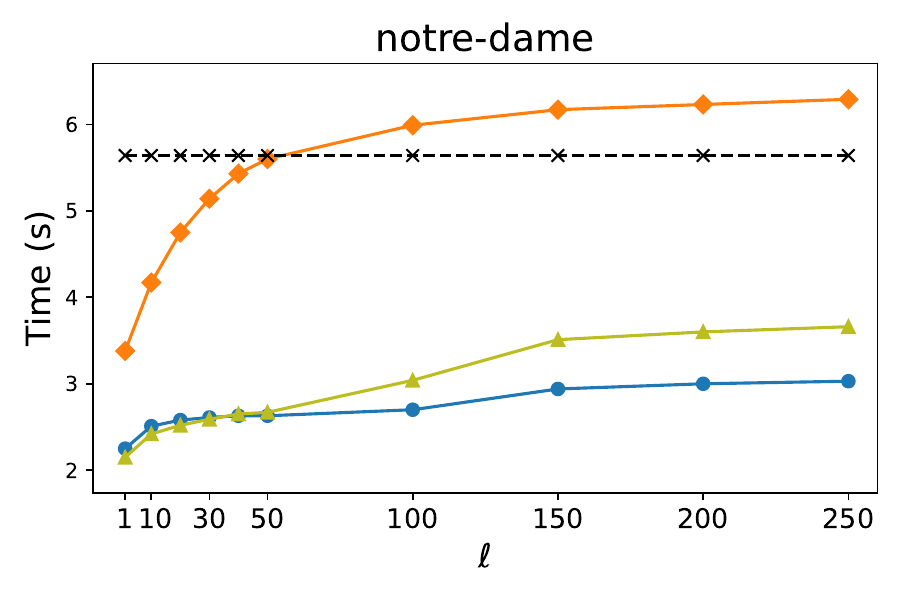}
    \includegraphics[width=0.32\textwidth]{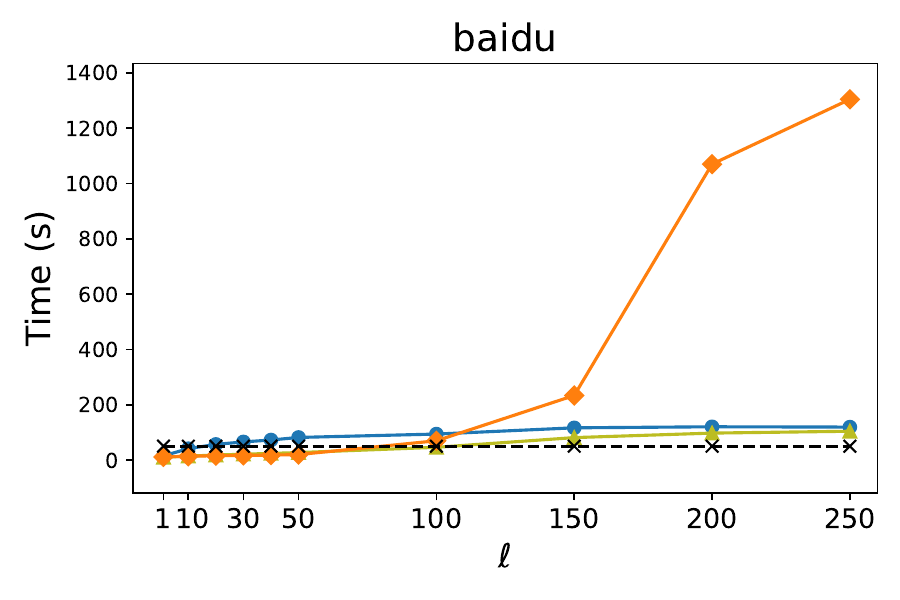}
    \includegraphics[width=0.32\textwidth]{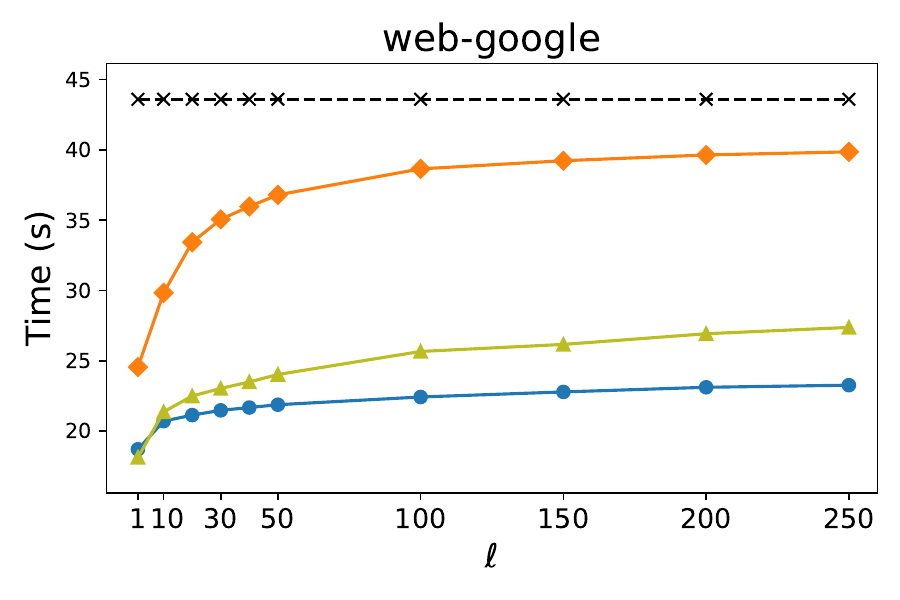} \\
    % Riga 6
    \includegraphics[width=0.32\textwidth]{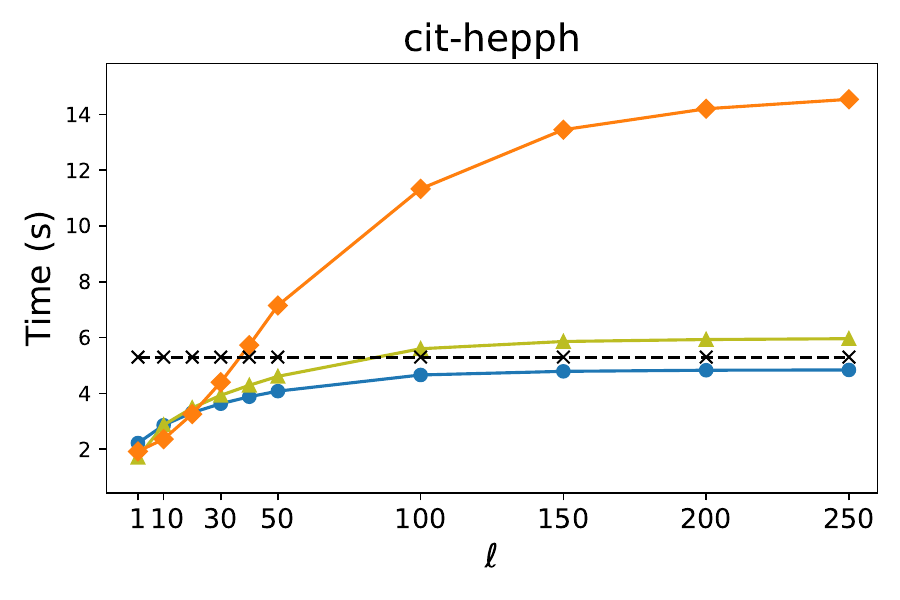}
    \includegraphics[width=0.32\textwidth]{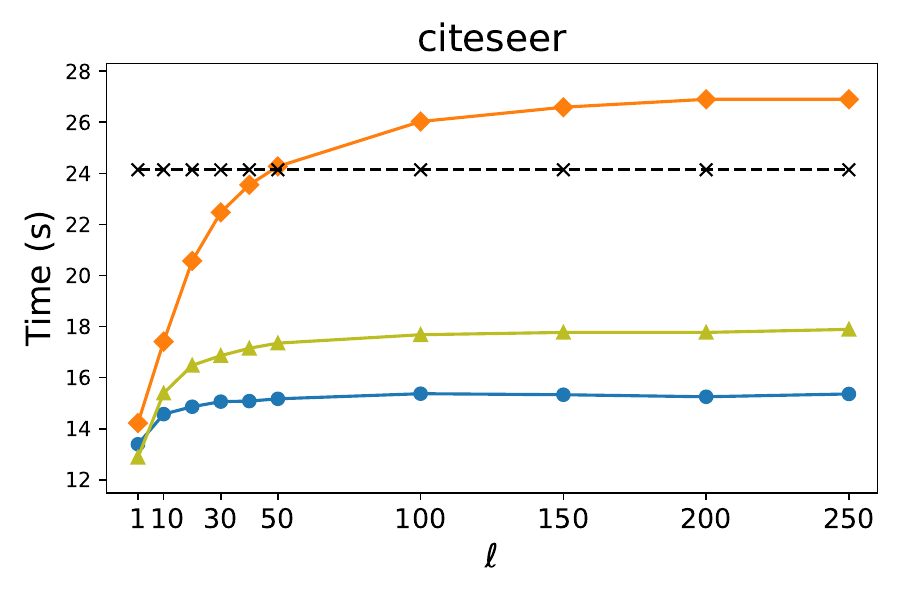} \\
    \includegraphics[width=0.5\textwidth]{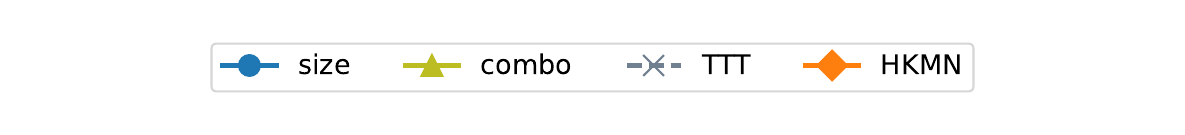}
    \hfill
    \caption{Comparison of \size, \combo, \hkmn, and \ttt\ on real-world graphs.}
    \label{fig:comparison-best}
\end{figure}

Interestingly, the picture is reversed in the right radar plot, which reports the running time normalized with respect to the slowest execution: \degeneracy, which is the most powerful pruning strategy, is also almost always the slowest, with the exception of {\small\tt baidu}, where \degree\ performs worst. The heuristics with weaker pruning, most notably \size, \softcore\ and \combo, are instead the fastest. Among these, \softcore\ and \combo\ behave very similarly, with \combo\ being slightly preferable in some cases. \degree\ exhibits an intermediate performance across benchmarks. Overall, the results highlight a clear trade-off between pruning power and efficiency: the running times observed in theory to perform the pruning tests (see \cref{sec:enumeration}) appear to have a stronger impact on performance than the actual amount of pruned nodes. Since running time is ultimately the main performance metric, in the remainder of this section we will focus on \size\ and \combo, which appear to be the most effective strategies.

%=============================================
\subsection{Comparison with state-of-the-art competitors}
\label{ss:competitors}

In~\cref{fig:comparison-best} we compare our best heuristics, \size\ and \combo, against \hkmn\ and \ttt, also analyzing how $\ell$ affects the running times. We chose a range of values for $\ell$ based on the observations in \cref{ss:distribution}, covering both the lower and the upper spectrum of the isolation factor.
A few trends emerge clearly from the charts. First, the running time of \ttt\ is independent of $\ell$, thus remaining constant. The two heuristics \size\ and \combo\ typically deliver the best performance over a large range of values of $\ell$, with runtimes that remain low and scale smoothly as $\ell$ increases, when larger portions of $\mathcal{T}_G$ need to be explored. Their curves are often close (see, e.g., {\small\tt livemocha} or {\small\tt cfinder-google}), \combo\ is marginally faster than \size\ on a few  instances (most notably, {\small\tt stanford} and {\small\tt italian-cnr}), but \size\ seems preferable in most of the other ones. 
On the other side, \hkmn\ is by far the slowest algorithm when $\ell$ gets large, with runtimes that grow quickly with $\ell$: this is due to an $\ell$-parameterized subroutine for finding minimal vertex covers. It can be instead competitive when $\ell$ is very small. \ttt\ is only occasionally competitive with \size\ and \combo, but still outperforms \hkmn\ for large $\ell$ except on {\small\tt web-google} and {\small\tt dblp}.
The observed relative differences between the algorithms persist across the datasets.
Even weak pruning heuristics that add a small test overhead on nodes of $\mathcal{T}_G$ translate into significant speedups, which can be as large as $4\times$ and $2\times$ with respect to \hkmn\ and \ttt, respectively.
% Experiments on synthetic graphs overall confirm the above analysis (see \cref{app:experiments-synthetic}).

\subsection{Experiments with synthetic networks}\label{app:experiments-synthetic}

We now report experiments on synthetic instances that mostly confirm the analysis on real-world graphs.

\begin{figure}[t]
    \centering
    \begin{subfigure}{\textwidth}
        \centering
        \includegraphics[width=1\textwidth]{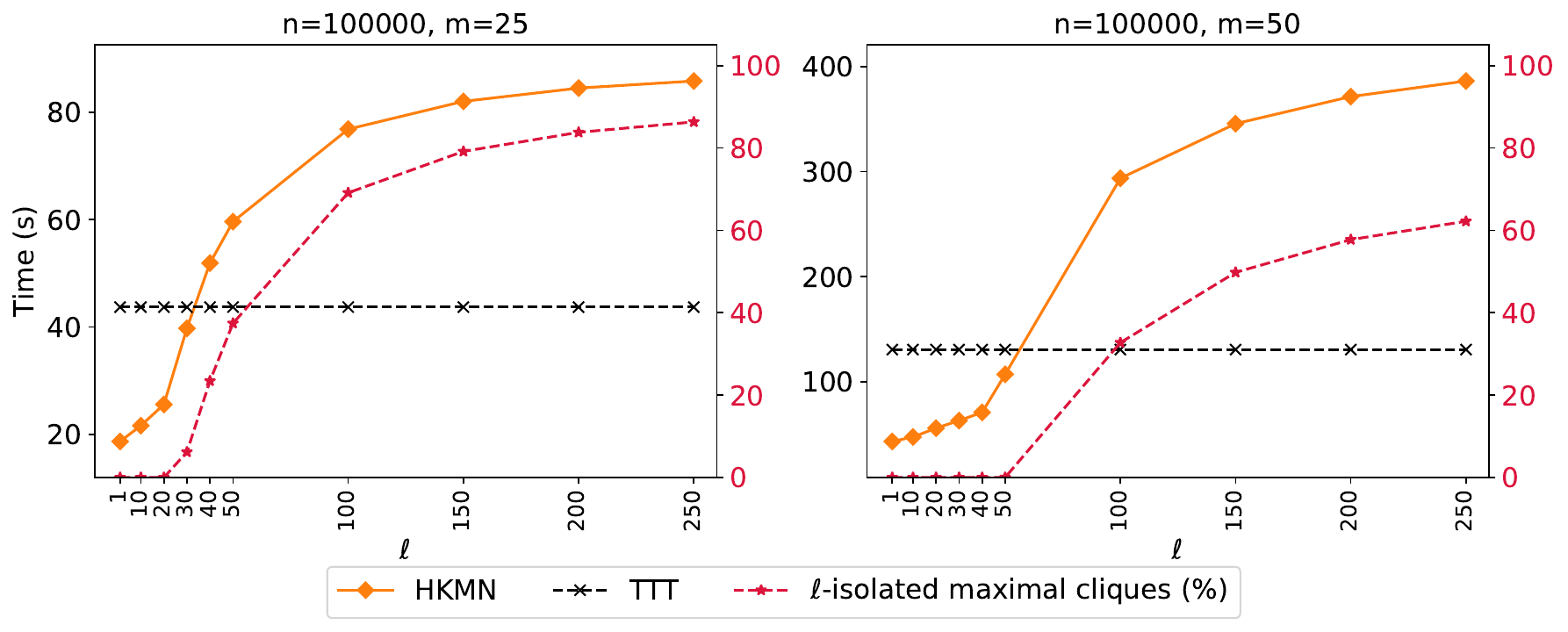}
        \caption{ Comparison between \hkmn\ and \ttt.}
        \label{fig:ba-comparison-a}
    \end{subfigure}
    \begin{subfigure}{\textwidth}
        \centering
        \includegraphics[width=1\textwidth]{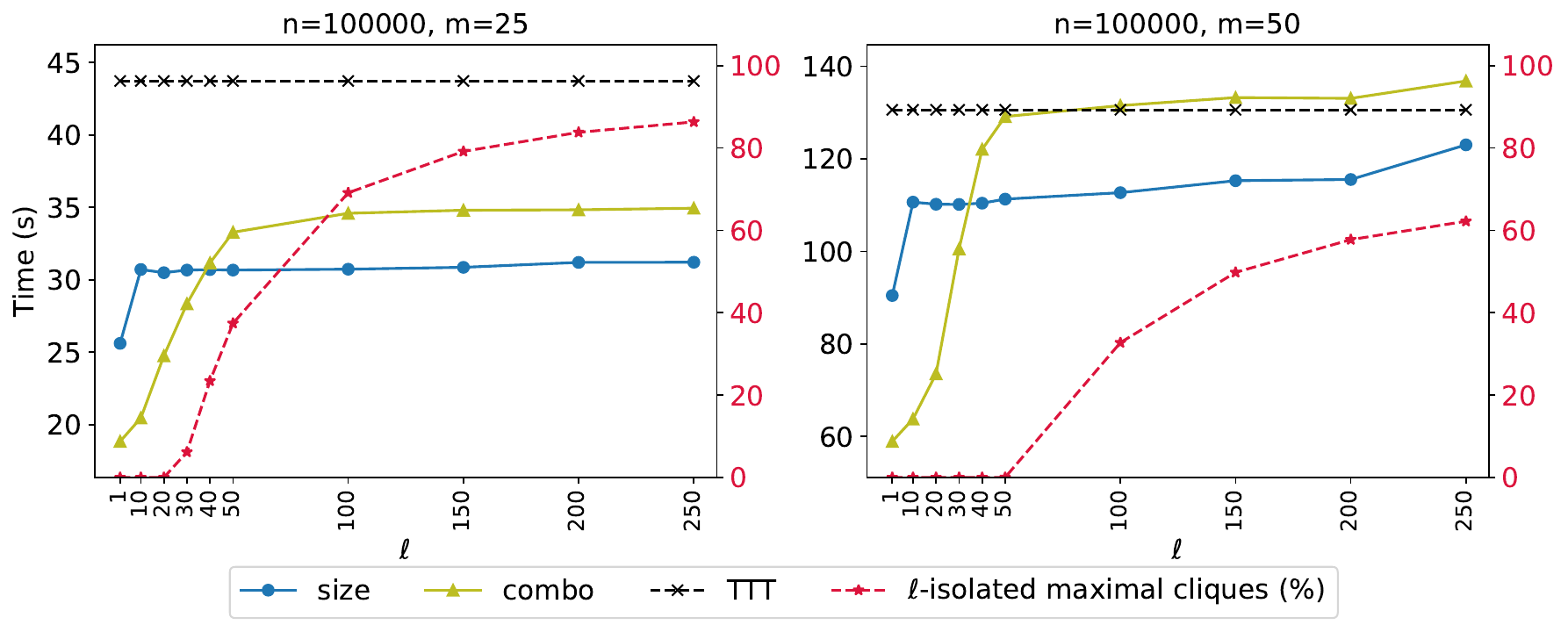}
        \caption{Comparison among \size, \combo, and \ttt.}
        \label{fig:ba-comparison-b}
    \end{subfigure}
    \caption{Results on BA$_{n,m}$ synthetic graphs. The dashed red line shows the number of $\ell$-isolated maximal cliques relative to the total number. 
    % (a) Comparison between \hkmn\ and \ttt. (b) Comparison among \size, \combo, and \ttt.
    }
    \label{fig:ba-comparison}
\end{figure}

We first discuss experiments on BA$_{n,m}$ networks.
% , confirming the results we observed on real-world instances.
\cref{fig:ba-comparison} compares \ttt\ against \hkmn\ (\cref{fig:ba-comparison-a}) and \size\ and \combo\ (\cref{fig:ba-comparison-b}), respectively. For $m=25$ and $\ell \leq 20$, we have that \hkmn\ is comparable with \combo\ and both of them perform better than \size. However, the dashed red line also shows that in the above range the $\ell$-isolated maximal cliques are very few (we also experimentally assessed their small size, akin to real-world networks). For higher values of $\ell$, the running time of \hkmn\ seems to be strongly dependent on the output size. Heuristic \size\ (and sometimes \combo), instead, provide a more efficient solution when $\ell \geq 30$. Curve trends are similar when $m=50$. The percentage of $\ell$-isolated cliques decreases: for $\ell=250$ it is barely above $60\%$ and remains close to zero for $\ell \leq 50$ (in this range \hkmn\ is the fastest algorithm). When $\ell>50$, \ttt\ is faster than \combo\ and \hkmn, but \size\ remains the preferred solution.

Networks generated according to the $G_{n,m,p}$ model have substantially different properties: they are typically very dense and $\ell$-isolated maximal cliques are barely present for small values of $\ell$ (e.g., $\ell < 80$). Conversely, for high values of $\ell$ (e.g., $\ell > 100$), their percentage becomes close to $100\%$, after a very steep transition (see the dashed red line in \cref{fig:gnmp-comparison}).
We reproduced the analysis of~\cite{DBLP:journals/tcs/HuffnerKMN09}, reporting in \cref{fig:gnmp-comparison} only the results with respect to~$\ell$.
Our results confirm those in~\cite{DBLP:journals/tcs/HuffnerKMN09}, with \hkmn\ outperforming other approaches for small values of $\ell$. \ttt\ seems instead the algorithm of choice for larger values. In any case, these graphs appear to represent worst-case scenarios for the algorithms proposed in this paper, although they remain very far from the properties of real-world networks.

\begin{figure}[thb]
\centering\includegraphics[width=1\textwidth]{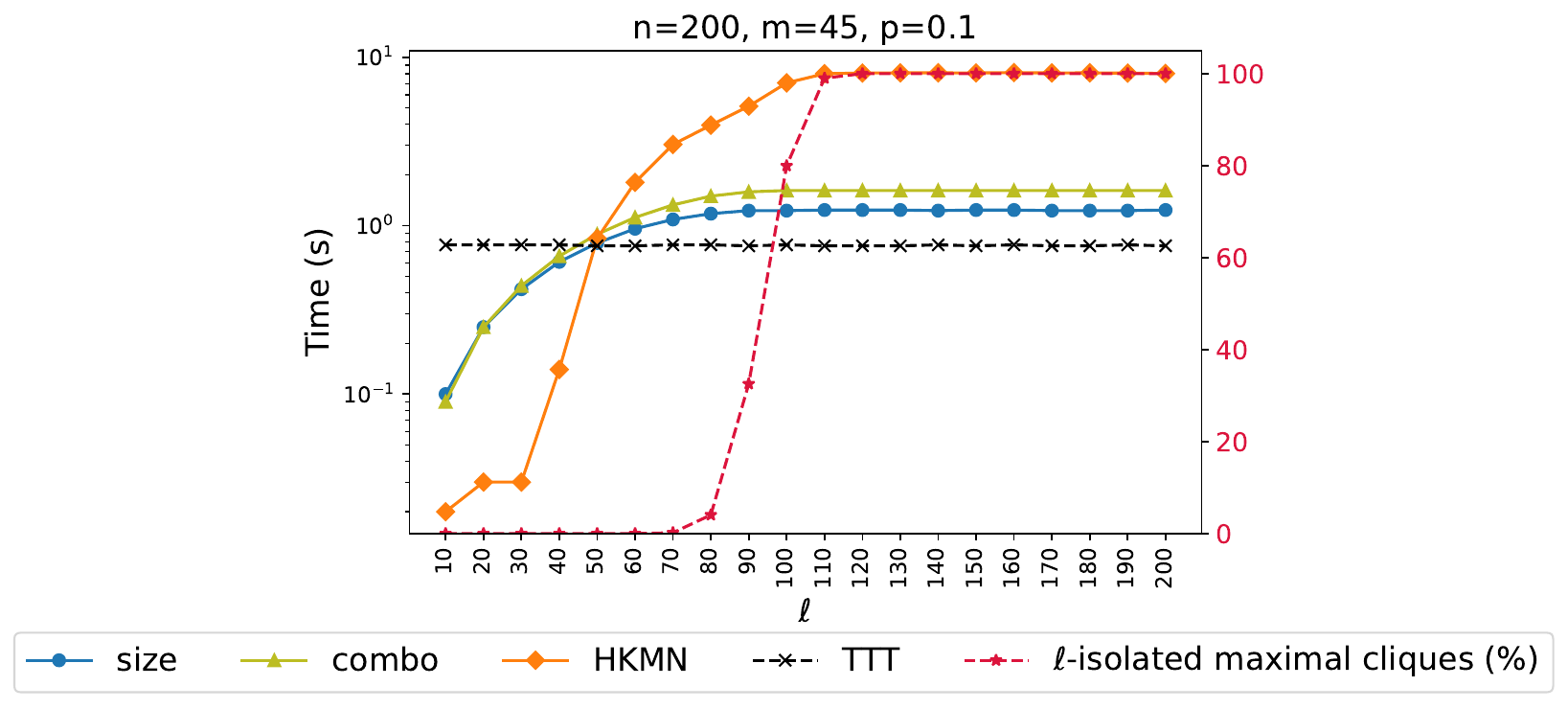}
%\vspace{-6mm}
\caption{Results on $G_{n,m,p}$ synthetic graphs. The dashed red line shows the number of $\ell$-isolated maximal cliques relative to the total number.}
\label{fig:gnmp-comparison}
\end{figure}

\section{Conclusions and open problems}\label{sec:conclusions}

This paper addressed the problem of enumerating isolated cliques, introducing pruning heuristics that can be applied on top of classical maximal cliques enumeration algorithms. We have theoretically proved correctness and conducted an extensive experimental analysis on a variety of benchmarks, proving the effectiveness of two of our variants against state-of-the-art competitors. 
Overall, $\ell$-isolation appears to be a useful lens for filtering cliques according to how strongly they are separated from the rest of the network. Small values of $\ell$ expose only the most tightly isolated structures, at the cost of excluding larger cliques. At the other extreme, too large values of $\ell$ may contribute little beyond the raw clique distribution, as almost every clique qualifies. The most informative patterns emerge at intermediate thresholds, such as $\ell=50$ or $100$, which strike a balance between selectivity and inclusiveness: they admit a substantial fraction of cliques, even larger ones, while still preserving the notion of relative isolation.
A natural way to address the tradeoff between selectivity and clique size, which is closely tied to the choice of $\ell$, could be to introduce alternative metrics in which the isolation factor is defined relative to clique size, rather than as an absolute threshold.

Isolation, and clique summarization more broadly, pose many other open problems (see also~\cite{DFP25}). In particular, it would be worthwhile to investigate whether our approach can be extended to the computation of $\ell$-isolated cliques by size (Problem 3 in~\cite{DFP25})
and to the stronger notion of max-$\ell$-isolation (see \cref{ss:l-isolation}). Another promising direction is to examine whether the problem admits fixed-parameter tractable algorithms under alternative parameters, such as degeneracy, either by refining the techniques of~\cite{DBLP:conf/esa/ItoIO05,DBLP:journals/talg/ItoI09} 
% ~\cite{DBLP:journals/talg/ItoI09}
or by focusing on specific graph families with distinctive structural properties.

\bibliography{references}

\end{document}